\documentclass[12pt,ruled]{article}

\usepackage{color}
\usepackage{verbatim}
\usepackage{mathtools}
\usepackage{url}
\usepackage{algorithm2e}
\usepackage{amsmath}
\usepackage{amssymb}
\usepackage{graphicx}
\usepackage[authoryear]{natbib}
\usepackage{bbm}
\usepackage{psfrag}
\usepackage{epsfig}
\usepackage{enumerate}

\usepackage{amsthm}
\usepackage{lscape}
\usepackage{indentfirst}
\usepackage{enumitem}
\usepackage{abstract}
\usepackage{fancyhdr}
\usepackage{multicol}
\usepackage{multirow}
\usepackage{autobreak} % break the math formula
\makeatletter

\theoremstyle{definition}

\newtheorem{prop}{Proposition}

\newcommand{\yujie}[1]{\color{blue}{#1}\color{black}}
\usepackage{graphicx}
\graphicspath{{figureSelected/}}

\title{ Rapid Detection of Hot-spot by  Tensor Decomposition on Space and Circular Time with  Application to Weekly Gonorrhea data }
\author{Yujie Zhao, Hao Yan, Sarah Holte, Yajun Mei}
\makeatother

\begin{document}
\maketitle

\begin{abstract}
	In many bio-surveillance and healthcare applications, data sources are measured from many spatial locations repeatedly over time, say, daily/weekly/monthly. In these applications, we are typically interested in detecting hot-spots, which are defined as some structured outliers that are sparse over the spatial domain but persistent over time. In this paper, we propose a tensor decomposition method to detect when and where the hot-spots occur. Our proposed methods represent the observed raw data as a three-dimensional tensor including a circular time dimension for daily/weekly/monthly patterns, and then decompose the tensor into three components: smooth global trend, local hot-spots, and residuals. A combination of LASSO and fused LASSO is used to estimate the model parameters, and a CUSUM procedure is applied to detect when and where the hot-spots might occur. The usefulness of our proposed methodology is validated through numerical simulation and a real-world dataset in the weekly number of gonorrhea cases from $2006$ to $2018$ for $50$  states in the United States.
\end{abstract}


\section{Introduction}

In many bio-surveillance and healthcare applications, data sources are measured from many spatial locations repeatedly over time, say, daily, weekly, or monthly. In these applications, we are typically interested in detecting {\it hot-spots}, which are defined as some structured outliers that are sparse over the spatial domain but persistent over time.
A concrete real-world motivating application  is the weekly number of gonorrhea cases from $2006$ to $2018$ for $50$  states in the United States, also see the detailed data description in the next section. From the monitoring viewpoint, there are two kinds of changes: one is the global-level trend, and the other is the local-level outliers. Here we are more interested in detecting the so-called hot-spots, which are local-level outliers with the following two properties: (1) spatial sparsity, i.e., the local changes are sparse over the spatial domain; and (2) temporal persistence, i.e., the local changes last for a reasonably long time period unless one takes some actions.

Generally speaking, the hot-spot detection can be thought as detecting sparse anomaly in spatio-temporal data, and there are three different categories of methodologies and approaches in the literature.
The first one is LASSO-based control chart that integrates LASSO estimators for change point detection and declares non-zero components of the LASSO estimators as the hot-spot, see \cite{LASSO}, \cite{LassoBased1}, \cite{vsaltyte2011spatial}.
Unfortunately, the LASSO-based control chart lacks the ability to separate the local hot-spots from the global trend of the spatio-temporal data.
The second category of methods is the dimension reduction based control chart where one monitors the features from PCA or other dimension reduction methods, see \cite{PCA}, \cite{tensorPCA1},  \cite{tensorPCA2}.
The drawback of PCA or other dimension reduction methods is that it fails to detect sparse anomalies and cannot take full advantage of the spatial location of hot-spot.
The third category of anomaly detection methods is the decomposition-based method that uses the regularized regression methods to separate the hot-spots from the background event, see \cite{AnomalyInVideo}, \cite{AnomalyInImage}, \cite{SSD}.
However, these existing approaches investigate structured images or curves data under the assumption that  the hot-spots are independent over the time domain.

In this paper, we propose a decomposition-based anomaly  detection method for spatial-temporal data when the hot-spots are autoregressive, which is typical for time series data. Our main idea is to represent the raw data as a $3$-dimensional tensor: states, weeks, years.
To be more specific, at each year, we observe a $50 \times 52$ data matrix that corresponds to $50$ states and $52$ weeks (we ignore the leap years).
Next, we propose to decompose the $3$-dimension tensor into three components: Smooth global trend, Sparse local hot-spot, and Residuals,
and term our proposed decomposition model as SSR-Tensor. When fitting the observed raw data to our proposed SSR-Tensor model, we develop a  penalized likelihood approach by adding two penalty functions: one is the LASSO type penalty to guarantee the sparsity of hot-spots,
and the other is the fused-LASSO type penalty for the autoregressive properties of hot-spots or time-series data.
By doing so, we are able to (1) detect when the hot-spots occur (i.e., the change point detection problem); and
(2) localize where and which type of the hot-spots occur (i.e., the spatial localization problem).

We would like to acknowledge that much research has been done on modeling and prediction of the spatio-temporal data.
Some popular time series models are AR, MA, ARMA model, etc., and the parameter can be estimated by   Yule-Walker method \citep{hannan1979determination}, maximum likelihood estimation or  least square method \citep{hamilton1994time}.
In addition, spatial statistics have also been extensively investigated on its own right, see  \citep{early_defination_neighbor,ecology,lan2004landslide,elhorst2014spatial,lots-of-ST-regression} for examples.
When one combines time series with spatial statistics, the corresponding spatio-temporal models generally become more complicated, see  \citep{ZhuJun,lai2015asymptotically,ST-model-book} for more discussions.

\yujie{
In principle, it is possible to represent the spatio-temporal process as a sequence of random vector  ${\bf Y}_{t}$ with weekly observation $t$, where ${\bf Y}$ is $p$-dimensional  vector that characterize the spatial domain (i.e., spatial dimension $p=50$ in our case study). 
%This approach will allow us to extend the time series models such as autoregressive from one-dimension to $p$-dimension, and thus can be useful for model fitting or prediction. 
% HaoYan: I feel this is very confusing since for autoregressive, we don't need the full covariance in the temporal domain, so I remove it
However, such an approach might not be computationally feasible in the context of hot-spot detection, in which one needs to specify the covariance structure of ${\bf Y}_{t}$, not only over the spatial domain, but also over the time domain. If we wrote all data into a vector, then the dimension of such vector is $50 \times 52 \times 13= 33,800$, and thus the covariance matrix is of dimension $33,800 \times 33,800,$ which is not computationally feasible. Meanwhile, under our proposed SSR-Tensor model, we essentially conduct a dimensional reduction by assuming that such a covariance matrix has a nice sparsity structure, as we reduce the dimensions $50, 52$ and $13$ to much smaller numbers, e.g., AR(1) model over the week or year dimension, and local correlation over the spatial domain.
}

It is useful to point out that while our paper focuses only on $3$-dimensional tensor due to our motivating application in gonorrhea, our proposed SSR-Tensor model can easily be extended to any $d$-dimensional tensor or data with $d \geq 3$, e.g., when we have further information, such as the unemployment rate, economic performance, and so on.  As the dimension  $d$ increases, we can simply add more corresponding bases, as our proposed model uses \textit{basis} to describe correlation within each dimension, and utilizes \textit{tensor product} for interaction between different dimensions. The capability of extending to high-dimensional data is one of the main advantages of our proposed SSR-Tensor model. Furthermore, our proposed SSR-Tensor model essentially involves block-wise diagonal covariation matrix, which allows ut to develop computationally efficient methodologies by using tensor decomposition algebra, see Section \ref{sec:computational_complexity} for more technical details.

The remainder of this paper is as follows. Section \ref{sec:data} discusses and visualizes the gonorrhea dataset, which is used as our motivating example and in our case study. Section \ref{sec:model_whole} presents our proposed SSR-Tensor model, and discusses how to estimate model parameters from observed data. Section \ref{sec:hot-spot_detection} describes how to use our proposed SSR-Tensor model to find hot-spots, both for temporal detection and for spatial localization. \yujie{Efficient numerical optimization algorithms are discussed in  Section \ref{sec:estimation}. }
Our proposed methods are then validated through extensive simulations in Section \ref{sec:simulation} and a case study in gonorrhea dataset in Section \ref{sec:case_study}.

\section{Data Description}
\label{sec:data}

\begin{figure}[t]
	\centering
	\begin{tabular}{ccc}
		\includegraphics[width = 0.3\textwidth]{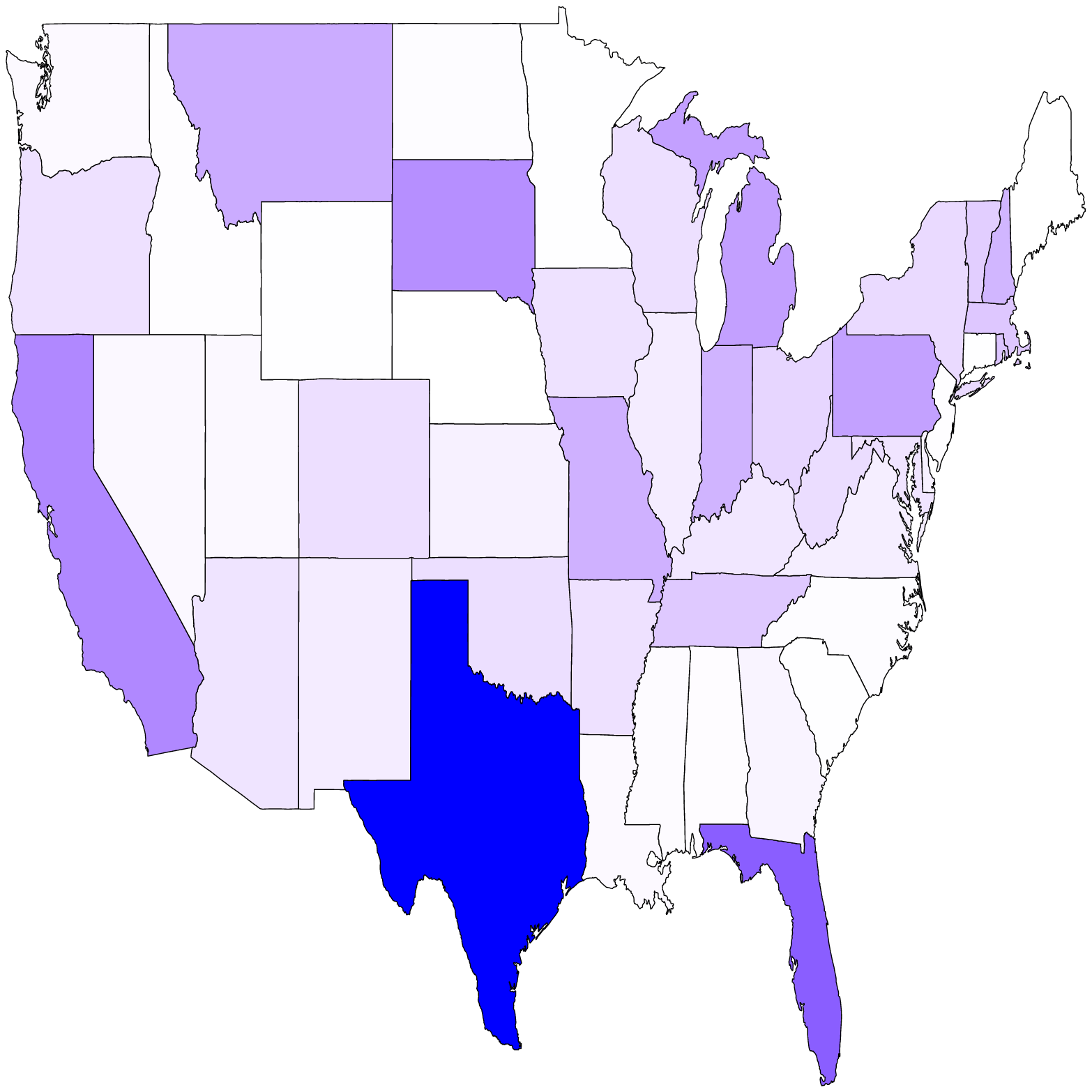}  &
		\includegraphics[width = 0.3\textwidth]{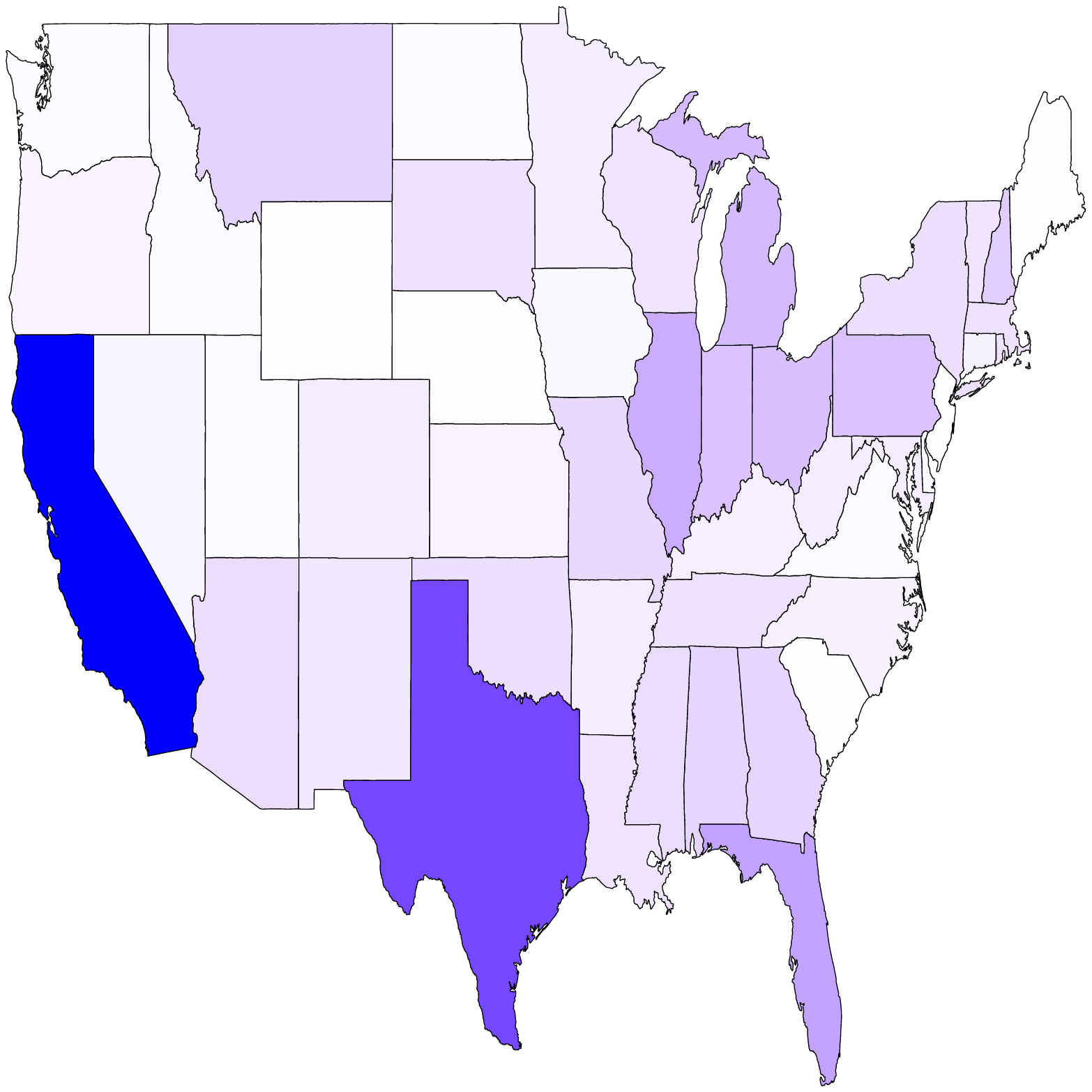}  &
		\includegraphics[width = 0.3\textwidth]{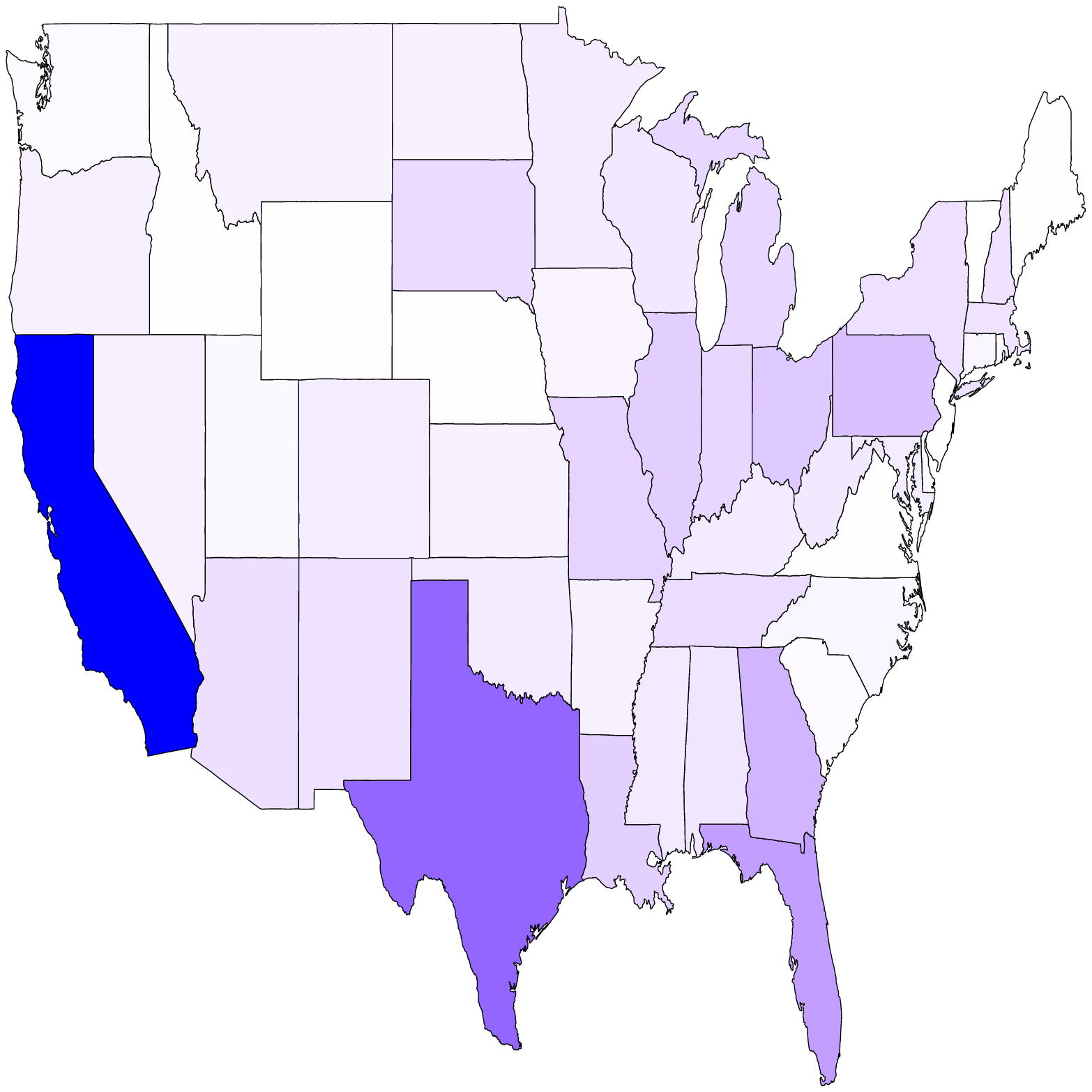}  \\
		week 1 & week 11  & week 21  \\
		\includegraphics[width = 0.3\textwidth]{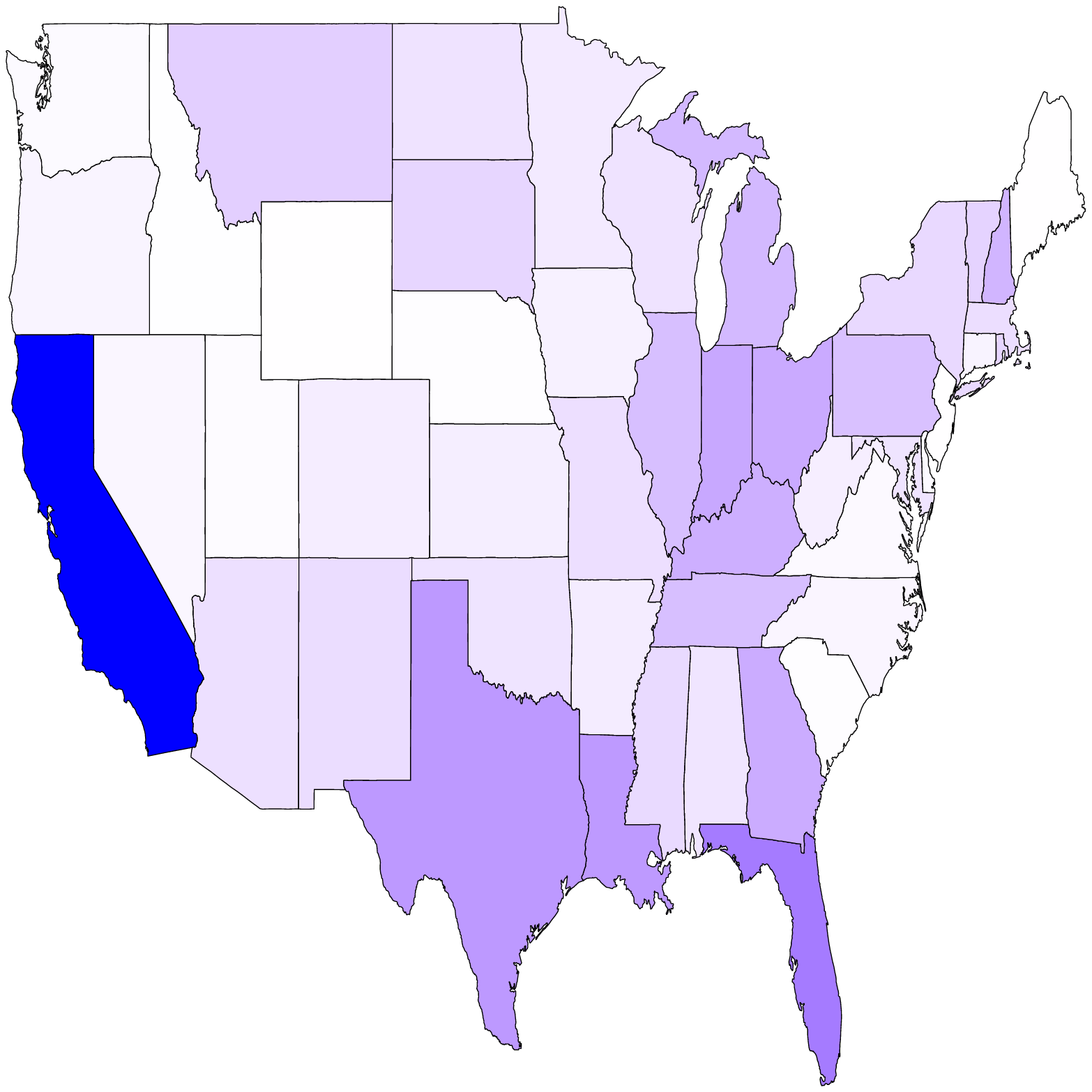}  &
		\includegraphics[width = 0.3\textwidth]{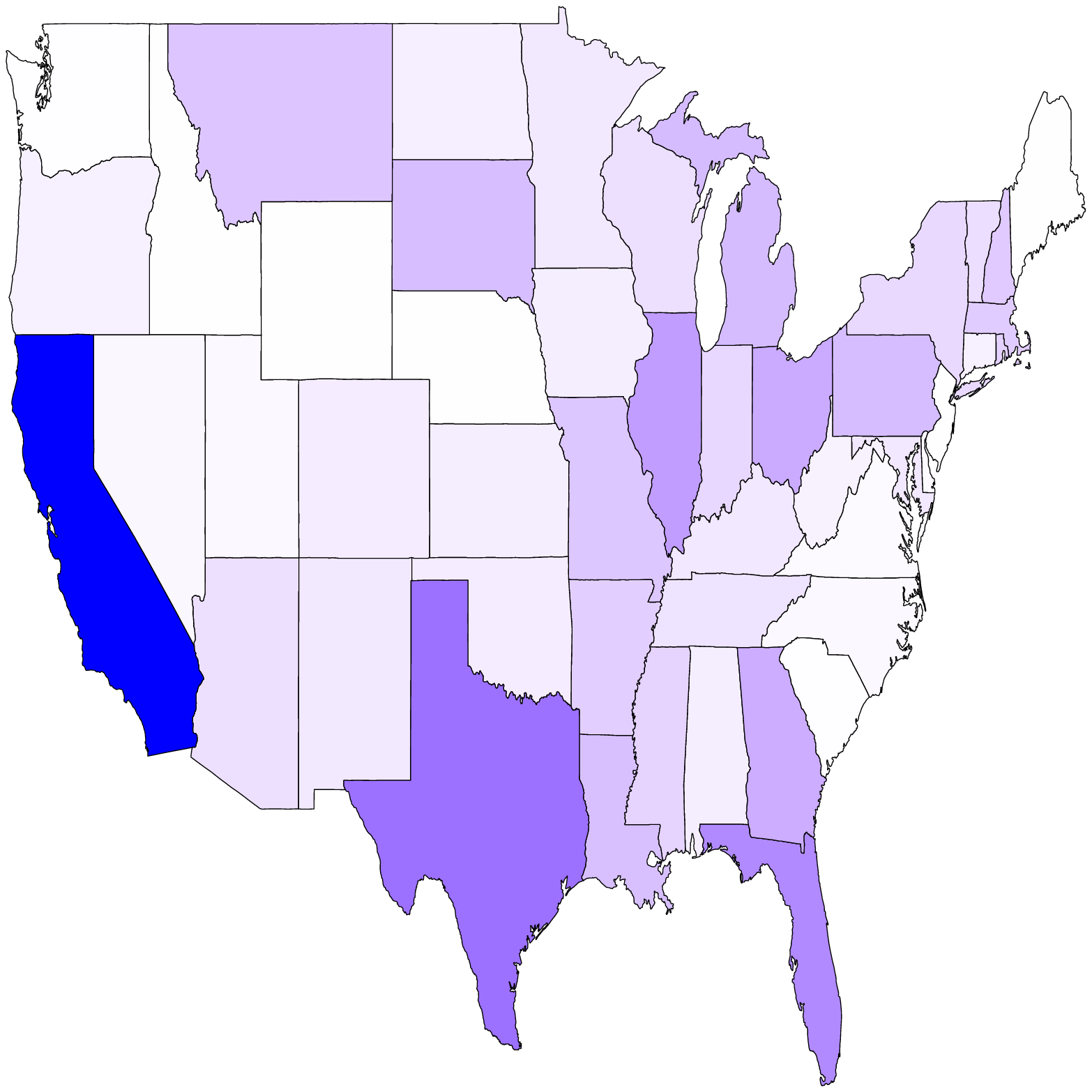}  &
		\includegraphics[width = 0.3\textwidth]{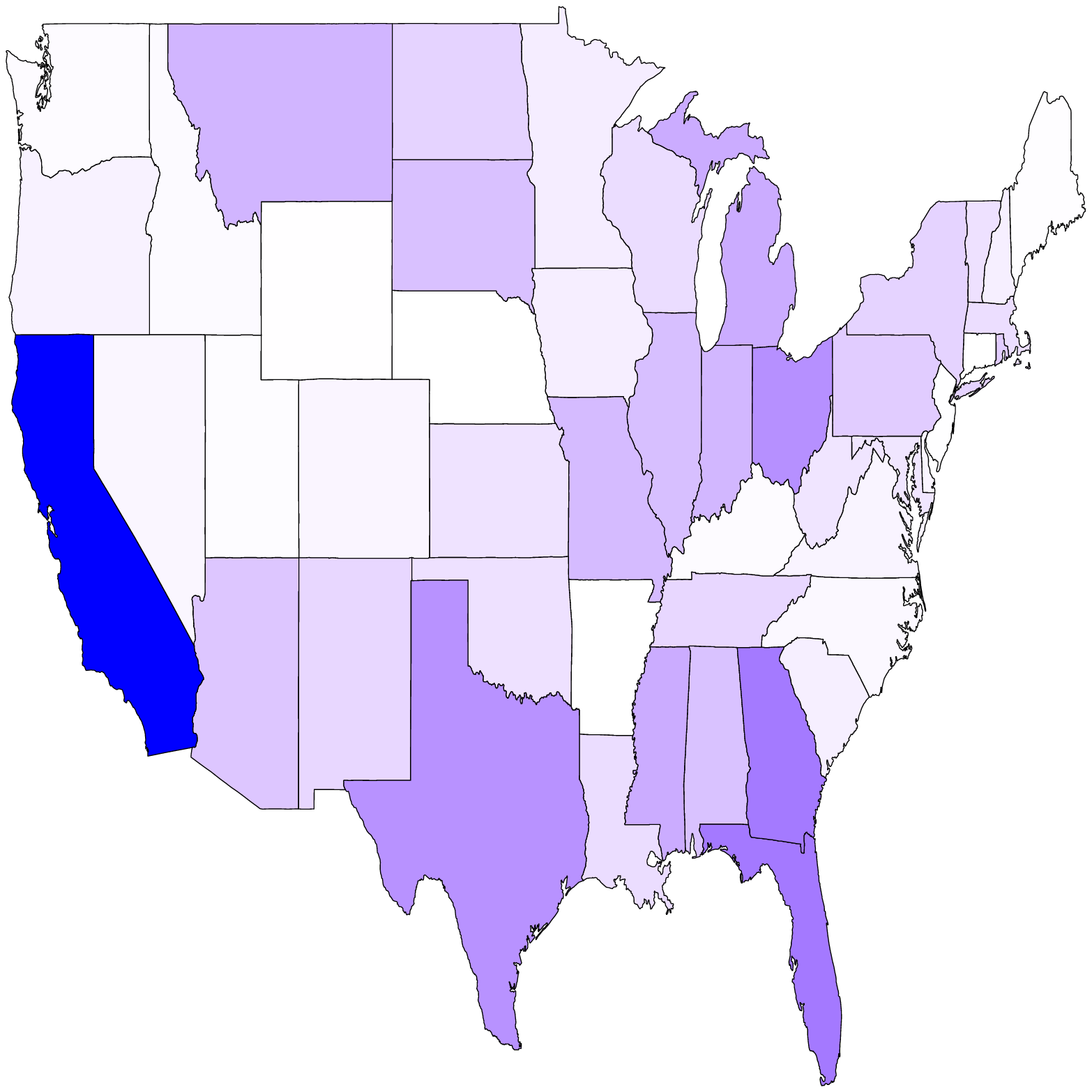}  \\
		week 31 & week 41  & week 51  \\
	\end{tabular}
	\caption{The cumulative number of gonorrhea cases at some selected weeks during years 2006-2018.
    \yujie{The deeper the color, the higher number of gonorrhea cases.}
    \label{fig:spatial_pattern_3}}
\end{figure}

To protect Americans from serious disease, the National Notifiable Disease Surveillance System (NNDSS) at the Centers for Disease Control and Prevention (CDC) helps public health monitor, control, and prevent about $120$ diseases, see its website \url{https://wwwn.cdc.gov/nndss/infectious-tables.html}.
One disease that receives intensive attention in recent years is gonorrhea, due to the possibility of multi-drug resistances.
Historically the instances of antibiotic resistance (in gonorrhea) have first been in the west and then move across the country.
\yujie{Since 1965, the CDC has collected the number of cumulative new infected patients every week in a calendar year. There are several changes on report policies or guidelines, and the latest one is year 2006. As a result, we focus on the weekly numbers of new gonorrhea patients during January 1, 2006 and December 31, 2018. The new weekly gonorrhea cases are computed as the difference of the cumulative cases in two consecutive weeks. The last week is dropped during this calculation.}
% HaoYan: I remove the last week deleted. In principal, you should use the new week of the next year minus the last week in the previous year. This is quite confusing so I remove it.

Let us first discuss the spatial patterns of the gonorrhea data  among 50 states.
For this purpose, we consider the cumulative number of gonorrhea cases from week 1 to week 52 by sum up all data during years 2006-2018. Figure
\ref{fig:spatial_pattern_3} plots some selected weeks (\#1, \#11, \#21, \#31, \#41, \#51).
\yujie
{
In Figure 1, if the state has a deeper and bluer color, then it experiences a higher number of gonorrhea cases.
}
One obvious pattern is that, California and Texas have generally higher number of gonorrhea cases as compared to other states.
In addition, the number of gonorrhea cases in the northern US  is smaller than that in the southern US.

Next, we consider the temporal pattern of the gonorrhea data set.
Figure \ref{fig:TimeSeriesAnnualUS} plots the annual number of gonorrhea cases over the years 2006-2018 in the US.
It is evident that there is a global-level  decreasing trend during 2010-2013.
One possible explanation is the Obamacare, which seems to reduce the risk of infectious diseases.
As we mentioned before, we are not interested in detecting this type of global changes,
and we focus on the detection of the changes on the local patterns, which are referred to as hot-spots in our paper.

\begin{figure}[t]
	\centering
	\includegraphics[width=0.5\textwidth]{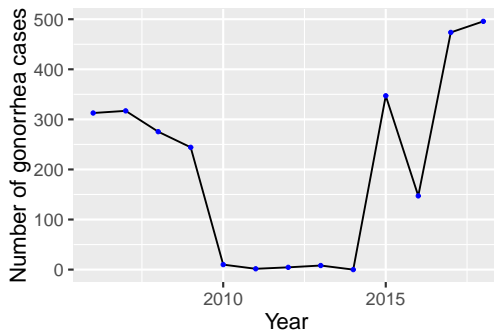}
	\caption{Annual number of gonorrhea cases (in thousands) over the years 2006-2018 in the US}
	\label{fig:TimeSeriesAnnualUS}
\end{figure}

Moreover, the gonorrhea data consists of weekly data, and thus it is necessary to address the circular patterns over the direction of ``week''. Figure  \ref{fig:circular_pattern} shows the country-scaled weekly gonorrhea case in the form of ``rose'' diagram for some selected years. In this figure, each direction represents a given week, and the length represents the number of gonorrhea cases for a given week. It reveals differences in the number of gonorrhea cases across a different week of the year. For instance, in July and August (in the direction of 8 o'clock on the circle), the number of gonorrhea case tends to be larger than other weeks.

\begin{figure}[htbp]
\centering
	\begin{tabular}{cc}
		\includegraphics[width = 0.5\textwidth]{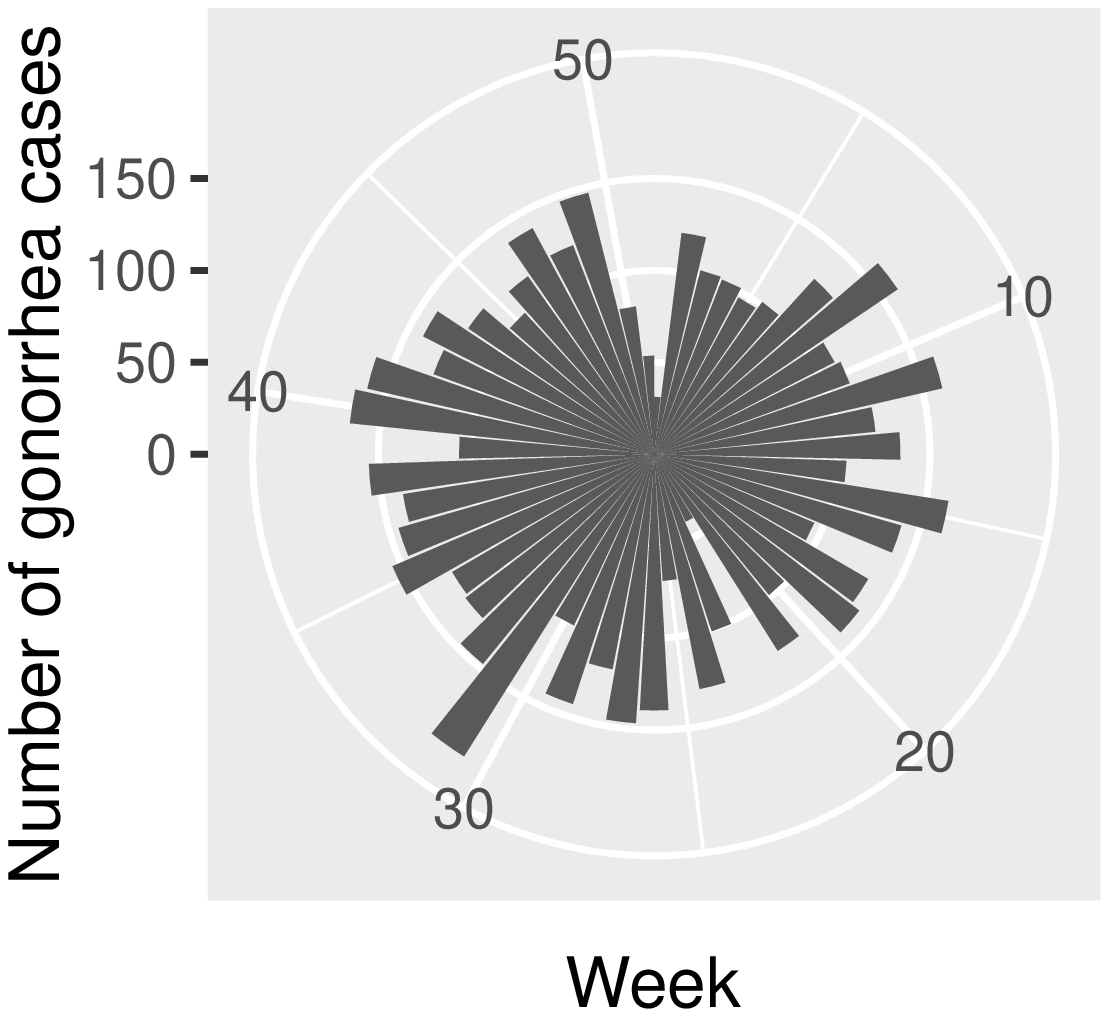}  &
		\includegraphics[width = 0.5\textwidth]{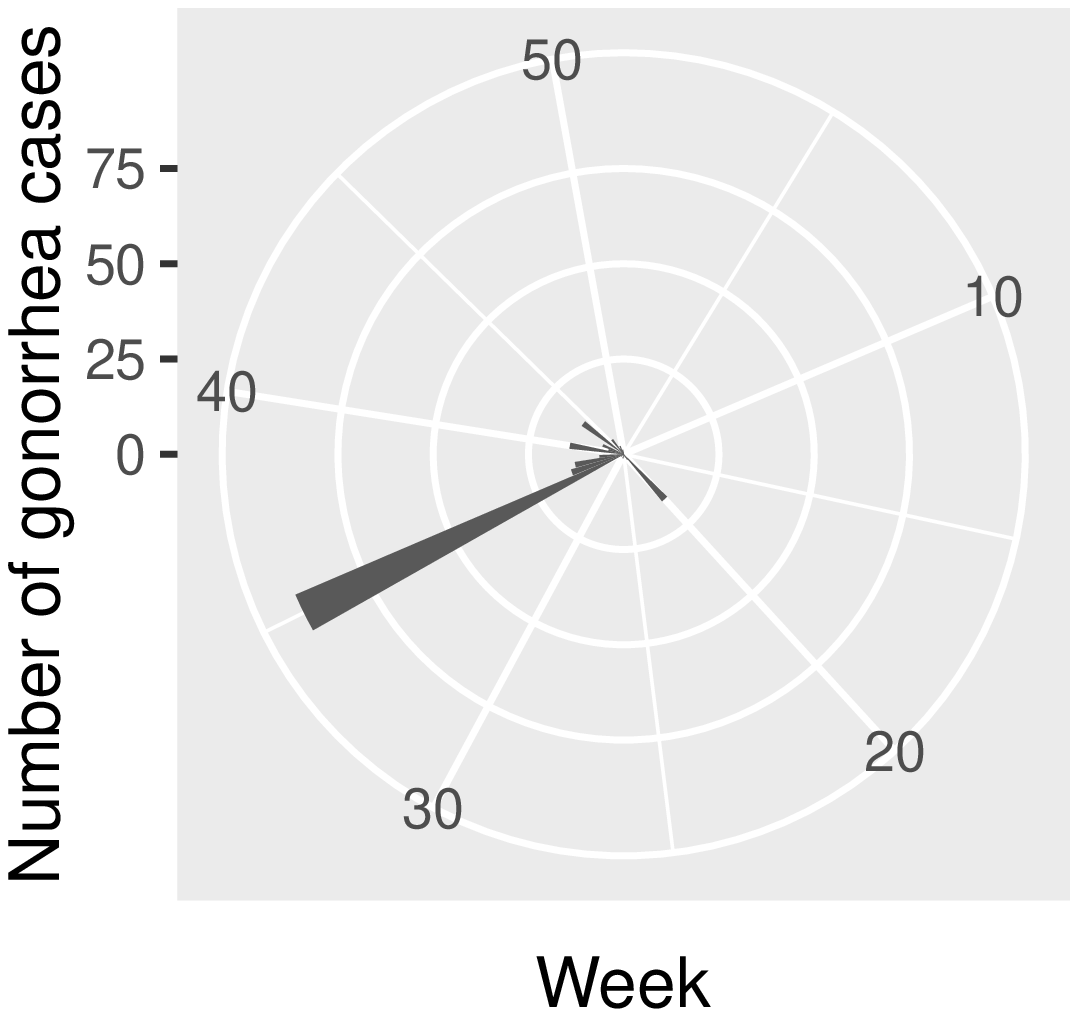}  \\
        2006  & 2010 \\
		\includegraphics[width = 0.5\textwidth]{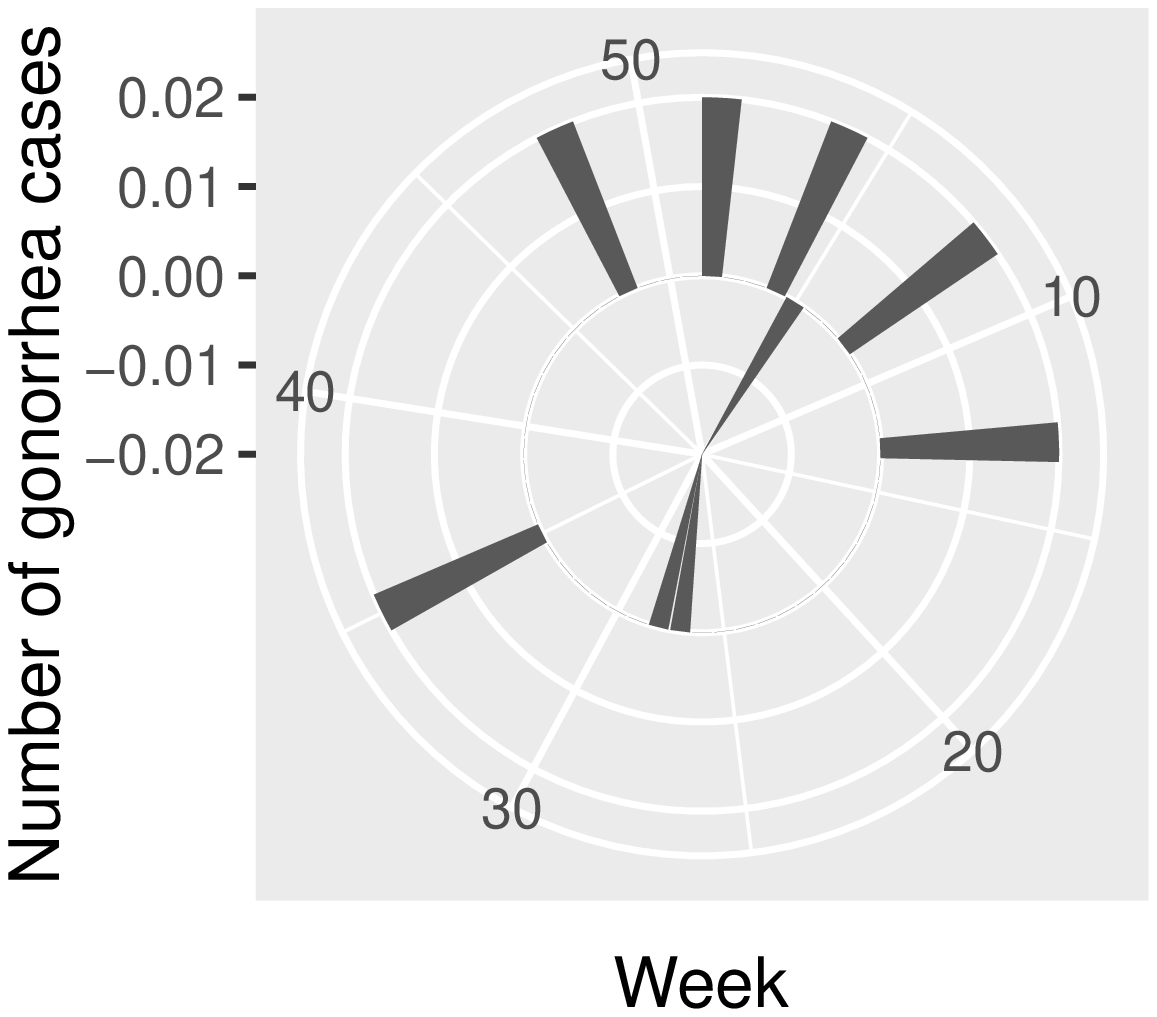}  &
		\includegraphics[width = 0.5\textwidth]{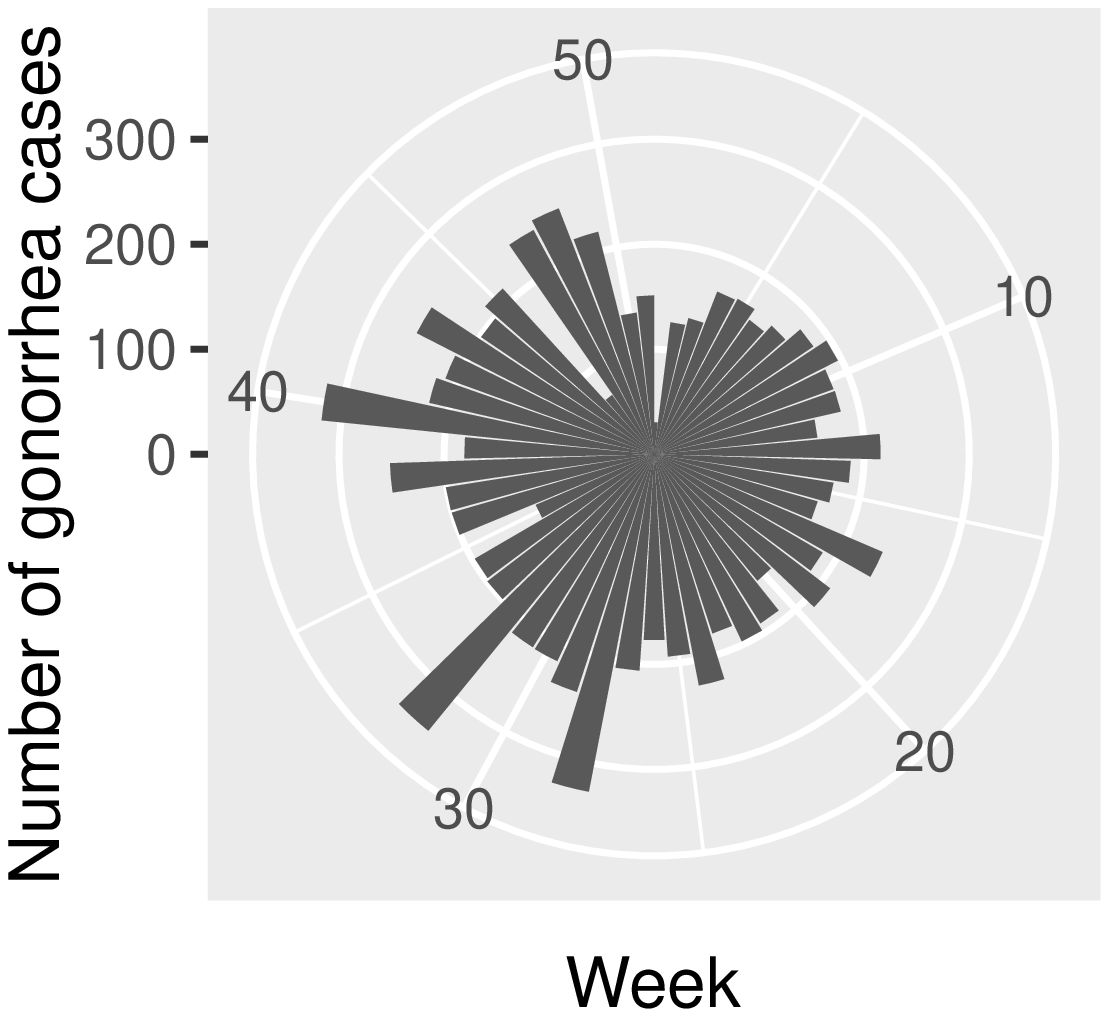}  \\
		 2014 & 2018
	\end{tabular}
	\caption{\yujie{Histograms of the number of gonorrhea cases of Year 2006, 2010, 2014, 2018. Each direction represents a given week, and the length represents the number of gonorrhea cases for a given week.}
		\label{fig:circular_pattern}}
\end{figure}

\section{Proposed Model}
\label{sec:model_whole}

In this section, we  present our proposed SSR-Tensor model, and postpone the discussion of hot-spot detection methodology to the next section. Owing to the fact that the gonorrhea data is of three dimensions, namely, \{state, week, year\}, it will likely have complex ``within-dimension'' and ``between-dimension'' interaction/correaltion relationship.
Within-dimension relationship includes within-state correlation, within-week correlation, and within-year correlation.
Between-dimension relationship includes between-state-and-week interaction, between-state-and-year interaction, as well as between-week-and-year interaction.
In order to handle these complex ``within'' and ``between'' interaction structures, we propose to use the tensor decomposition method, where bases are used to address ``within-dimension'' correlation, and the tensor product is used for ``between-dimension'' interaction.
Here, the basis is a very important concept where different basis can be chosen for different dimensions.
\textcolor[rgb]{0.00,0.07,1.00}
{
Detailed discussions of the choice of bases are presented in Section \ref{sec:hot-spot_detection_performance}.
}

For the convenience of notation and easy understanding, we first introduce some basic tensor algebra and notation in Section \ref{sec:Tensor_Algebra_and_Notation}.
Then Section \ref{sec:model} presents our proposed model that is able to characterize the complex correlation structures.

\subsection{Tensor Algebra and Notation } \label{sec:Tensor_Algebra_and_Notation}

In this section, we introduce basic notations, definitions, and operators in tensor (multi-linear) algebra that are useful in this paper.
Throughout the paper,
scalars are denoted by lowercase letters (e.g., $\theta$),
vectors are denoted by lowercase boldface letters ($\boldsymbol{\theta}$),
matrices are denoted by uppercase boldface letter ($\boldsymbol{\Theta}$),
and tensors by curlicue letter ($\vartheta$).
For example, an order-$N$ tensor is represented by
$
\vartheta \in \mathbb{R}^{I_{1} \times \cdots \times I_{N}}
$,
where $I_{k}$ represent the mode-$n$ dimension of $\vartheta$
for $k = 1, \ldots, N$.

\yujie{
The mode-$n$ product of a tensor
$
\vartheta \in \mathbb{R}^{I_{1} \times \ldots \times I_{N}}
$
by a matrix $\mathbf{B} \in \mathbb{R}^{J_{n}\times I_{n}}$ is a tensor
$
\mathcal{A} \in \mathbb{R}^{ I_{1} \times \ldots I_{n-1} \times J_n \times I_{n+1} \times \ldots I_{N} }
$,
denoted as
$
\mathcal{A} = \vartheta \times_n \mathbf{B},
$
where each entry of $\mathcal{A}$ is defined as the sum of products of corresponding entries in $\mathcal{A}$ and $\mathbf{B}$:
$
\mathcal{A}_{i_1,\ldots, i_{n-1},j_{n},i_{n+1}, \ldots, i_N}
=
\sum_{i_{n}}
\vartheta_{i_1, \ldots, i_{N}}
\mathbf{B}_{j_n,i_n}
$.
Here we use the notation $\mathbf{B}_{j_n,i_n}$ to refer the $(j_n, i_n)$-th entry in matrix $\mathbf{B}$.
 The notation $\vartheta_{i_1, \ldots, i_{N}}$ is used to refer to the entry in tensor $\vartheta$ with index $(i_1, \ldots, i_{N})$.
 The notation
$
\mathcal{A}_{i_1,\ldots, i_{n-1},j_{n},i_{n+1}, \ldots, i_N}
$
is used to refer the entry in tensor
$\mathcal{A}$ with index
$(i_1,\ldots, i_{n-1},j_{n},i_{n+1}, \ldots, i_N)$.

The mode-n unfold of the tensor $\vartheta \in \mathbb{R}^{I_{1} \times \ldots \times I_{N}}$ is denoted by
$
\vartheta_{(n)} \in \mathbb{R}^{I_n \times (I_1\times \ldots I_{n-1} \times I_{n+1} \times I_N)},
$
where the column vector of $\vartheta_{(n)}$ are the mode-n vector of $\vartheta$.
The mode-n vector of $\vartheta$ are defined as the $I_n$ dimensional vector obtained from $\vartheta$ by varying the index $i_n$ while keeping all the other indices fixed.
For example, $\vartheta_{:,2,3}$ is a model-1 vector.

A very useful technique in the tensor algebra is the Tucker decomposition, which decomposes a tensor into a core tensor multiplied by matrices along each mode:
$
\mathcal{Y}
=
\vartheta \times_{1} \mathbf{B}^{(1)}\times_{2}\mathbf{B}^{(2)}\cdots\times_{N}\mathbf{B}^{(N)}
$,
where $\mathbf{B}^{(n)}$ is an orthogonal $I_{n}\times I_{n}$ matrix and is a principal component mode-$n$ for $n =1, \ldots, N$.
Tensor product can be represented equivalently by a Kronecker product, i.e.,
$
\mathrm{vec}(\mathcal{Y})
=
(\mathbf{B}^{(N)} \otimes \cdots \otimes \mathbf{B}^{(1)}) \mathrm{vec} (\vartheta)
$,
where $\mathrm{vec}(\cdot)$ is the vectorized operator.
Finally, the definition of Kronecker product is as follow: Suppose $\mathbf{B}_{1}\in\mathbb{R}^{m \times n}$ and $\mathbf{B}_{2}\in\mathbb{R}^{p\times q}$ are matrices, the Kronecker product of these matrices, denoted by $\mathbf{B}_{1}\otimes\mathbf{B}_{2}$,
is an $mq\times nq$ block matrix defined by
$$
\mathbf{B}_{1}\otimes\mathbf{B}_{2}
=
\left[\begin{array}{ccc}
b_{11}\mathbf{B}_2 & \cdots & b_{1n}\mathbf{B}_2 \\
\vdots             & \ddots & \vdots             \\
b_{m1}\mathbf{B}_2 & \cdots & b_{mn}\mathbf{B}_2
\end{array}\right].
$$
}

\subsection{Our Proposed SSR-Tensor Model} \label{sec:model}

Our proposed SSR-Tensor model is built on tensors of order three, as it is inspired by the gonorrhea data,  which can be represented as a three-dimension tensor $\mathcal{Y}_{n_{1}\times n_{2}\times T}$ with $n_1=50$ states, $n_2=51$ weeks, and $T=13$ years.
Note that the $i$-th$, $j$-th$, and $k$-th slice of the 3-D tensor along the dimension of state, week, and year can be achieved as $\mathcal{Y}_{i::},\mathcal{Y}_{:j:},\mathcal{Y}_{::k}$ correspondingly, where $i=1\cdots n_{1}$, $j=1\cdots n_{2}$ and $k=1\cdots T$.
For simplicity, we denote $\mathbf{Y}_{k}=\mathcal{Y}_{::k}$.
We further denote $\mathbf{y}_{k}$ as the vectorized form of $\mathbf{Y}_{k}$, and $\mathbf{y}$ as the vectorized form of $\mathcal{Y}$.

The key idea of our proposed model is to separate the global trend from the local pattern by decomposing the tensor $\mathbf{y}$ into three parts, namely the smooth global trend $\boldsymbol{\mu}$, local hot-spot $\mathbf{h}$, and residual $\mathbf{e}$, i.e. $\mathbf{y}=\boldsymbol{\mu}+\mathbf{h}+\mathbf{e}$.
For the first two of the components (e.g. the global trend mean and local hot-spots), we introduce basis decomposition framework to represent the structure of the within correlation in the global background and local hot-spot, also see  \citet{SSD}.

To be more concrete, we assume that global trend mean and local hot-spot can be represented as $\boldsymbol{\mu}=\mathbf{B}_{m}\boldsymbol{\theta}_{m}$ and $\boldsymbol{h}=\mathbf{B}_{h}\boldsymbol{\theta}_{h}$,
where  $\mathbf{B}_{m}$ and $\mathbf{B}_{h}$ are two  bases that will discussed below, and $\boldsymbol{\theta}_{m}$ and $\boldsymbol{\theta}_{h}$ are the model coefficients vector of length $n_{1}n_{2}T$ and needed to be estimated (see Section \ref{sec:estimation}).
Here  the subscript of \textit{m} and \textit{h} are abbreviations for mean and hot-spot.
Next,  it is useful to discuss how to choose the bases $\mathbf{B}_{m}$ and $\mathbf{B}_{h},$ so as to characterize  the complex ``within'' and ``between'' correlation or interaction structures.
For the ``within" correlation structures, we propose to use pre-specified bases, $\mathbf{B}_{m,s}$ and $\mathbf{B}_{h,s}$, for within-state correlation in global trend and hot-spot, where the subscript of \textit{s} is an abbreviation for states.
Similarly,  $\mathbf{B}_{m,w}$ and $\mathbf{B}_{h,w}$ are the pre-specified bases for within-correlation of the same week, whereas $\mathbf{B}_{m,y}$ and $\mathbf{B}_{h,y}$ are the bases for within-time correlation over time.
As for the ``between'' interaction, we use tensor product to describe it, i.e, $\mathbf{B}_{m}=\mathbf{B}_{m,s}\otimes\mathbf{B}_{m,w}\otimes\mathbf{B}_{m,y}$ and $\mathbf{B}_{h}=\mathbf{B}_{h,s}\otimes\mathbf{B}_{h,w}\otimes\mathbf{B}_{h,y}$.
This Kronecker product has been proved to have better computational efficiency in the tensor response data \cite{TensorAlgebra}.
\textcolor[rgb]{0.00,0.07,1.00}
{
Mathematically speaking, all these bases are matrices, which is pre-assigned in our paper.
And the choice of bases in shown in Section \ref{sec:hot-spot_detection_performance}.
}
With the well-structured ``within'' and ``between'' interaction, our proposed model can be written as:
\begin{equation}
\mathbf{y}=(\mathbf{B}_{m,s}\otimes\mathbf{B}_{m,w}\otimes\mathbf{B}_{m,y})\boldsymbol{\theta}_{m}+(\mathbf{B}_{h,s}\otimes\mathbf{B}_{h,w}\otimes\mathbf{B}_{h,y})\boldsymbol{\theta}_{h}+\mathbf{e},
\label{equ:model}
\end{equation}
where $\mathbf{e}{\sim}N(0,\sigma^{2}\mathbf{I})$ is the random noise.
Mathematically speaking, both $\mathbf{B}_{m,s}$ and $\mathbf{B}_{h,s}$ are $n_{1}\times n_{1}$ matrix, $\mathbf{B}_{m,w}$ and $\mathbf{B}_{h,w}$ are $n_{2}\times n_{2}$ matrix and $\mathbf{B}_{m,y}$ and $\mathbf{B}_{h,y}$ are $T \times T$ matrix, respectively.

Mathematically, our proposed model in \eqref{equ:model} can be rewritten into a tensor format:
\begin{equation}
\mathcal{Y}=
\vartheta_{m} \times_{3}
\mathbf{B}_{m,y} \times_{2}
\mathbf{B}_{m,w} \times_{1}
\mathbf{B}_{m,s}
+
\vartheta_{h}\times_{3}
\mathbf{B}_{h,y} \times_{2}
\mathbf{B}_{h,w}\times_{1}
\mathbf{B}_{h,s}+\mathbf{e},
\label{equ:meaning_of_theta}
\end{equation}
where $\vartheta_{m}$ and $\vartheta_{h}$ is the tensor format of $\boldsymbol{\theta}_{m}$ and $\boldsymbol{\theta}_{h}$ with dimensional  $n_{1}\times n_{2}\times T$.
Accordingly, the $((k-1)n_{1}n_{2}+(i-1)n_{1}+j)$-th entry of $\boldsymbol{\theta}_{h}$, $\boldsymbol{\theta}_{m}$ can estimate the global mean and hot-spot in $i$-th state and $j$-th week in $k$-th year respectively.
The tensor representation in equation \eqref{equ:meaning_of_theta} allows us to develop computationally efficient methods for estimation and prediction.

\subsection{Estimation of Hot-spots}

With the proposed SSR-Tensor model above, we can now discuss the estimation of hot-spot parameters $\boldsymbol{\theta}$'s (including $\boldsymbol{\theta}_m$, $\boldsymbol{\theta}_h$) in our model in \eqref{equ:model} or \eqref{equ:meaning_of_theta}
from the data via the penalized likelihood function.
We propose to add two penalties in our estimation. First, because hot-spots rarely occur, we assume that $\boldsymbol{\theta}_{h}$ is sparse and the majority of entries in the hot-spot coefficient $\boldsymbol{\theta}_{h}$ are zeros. Thus we propose to add the penalty $R_{1}(\boldsymbol{\theta}_{h})=\lambda\Vert\boldsymbol{\theta}_{h}\Vert_{1}$  to encourage the sparsity property of $\boldsymbol{\theta}_{h}$.
Second, we assume there is temporal continuity of the hot-spots, as
the usual phenomenon of last year is likely to affect the performance of hot-spot in this year.
Thus, we add the second penalty
$ R_{2}(\boldsymbol{\theta}_{h})= \lambda_ 2\Vert \mathbf{D} \boldsymbol{\theta}_{h} \Vert_1$
to ensure the yearly continuity of the hot-spot, where
$
\mathbf{D} =   \mathbf{D}_{s} \otimes  \mathbf{D}_{w} \otimes  \mathbf{D}_{y}
$
with $ \mathbf{D}_{s}$ as identical matrix of dimension $n_1\times n_1$, and $T \times T$ matrix
$
\mathbf{D}_{y} =
\left[
\begin{array}{ccccc}
1 & -1\\
&  & \ddots & \ddots\\
&  &  & 1 & -1\\
&  &  &  & 1
\end{array}
\right]
$,
$n_2 \times n_2$ matrix
$
\mathbf{D}_{w} =
\left[
\begin{array}{ccccc}
1 & -1\\
&  & \ddots & \ddots\\
&  &  & 1 & -1\\
-1&  &  &  & 1
\end{array}
\right].
$
With the formula of $\mathbf{D}_y$, the hot-spot has the property of yearly continuity.
By the formula of $\mathbf{D}_w$, the hot-spot has a weekly circular pattern.

By combining both penalties, we propose to estimate the parameters via the following optimization problem:
\begin{eqnarray}
\label{equ:eatimation}
&& \arg\min_{\boldsymbol{\theta}_{m},\boldsymbol{\theta}_{h}}
\Vert\boldsymbol{e}\Vert^{2}
+
\lambda_{1}\Vert\boldsymbol{\theta}_{h}\Vert_{1}
+
\lambda_{2}\Vert \mathbf{D}\boldsymbol{\theta}_{h}\Vert_{1}\\
&& \mbox{subject to} \;\;
\boldsymbol{y} =
(\mathbf{B}_{m,s}\otimes\mathbf{B}_{m,w}
\otimes
\mathbf{B}_{m,y})\boldsymbol{\theta}_{m}
+
(\mathbf{B}_{h,s}
\otimes
\mathbf{B}_{h,w}
\otimes
\mathbf{B}_{h,y})\boldsymbol{\theta}_{h}
+\mathbf{e}, \nonumber
\end{eqnarray}
where
$
\boldsymbol{\theta}_{m} =
\mathrm{vec}(\boldsymbol{\theta}_{m,1},\ldots, \boldsymbol{\theta}_{m,t},\ldots,\boldsymbol{\theta}_{m,T})$
and $\boldsymbol{\theta}_{h} =
\mathrm{vec}(\boldsymbol{\theta}_{h,1},,\ldots, \boldsymbol{\theta}_{h,t},\ldots, \boldsymbol{\theta}_{h,T})$.
\textcolor[rgb]{0.00,0.07,1.00}
{
The choice of the turning parameters $\lambda_1, \lambda_2$ will be discussed in Section \ref{sec:hot-spot_detection}.
}

Note that there are two penalties in equation \eqref{equ:eatimation}:
$\lambda_{1}\Vert\boldsymbol{\theta}_{h}\Vert_{1}$ is the LASSO penalty
to control both the sparsity of the hot-spots and $\lambda_{2}\Vert \mathbf{D}\boldsymbol{\theta}_{h}\Vert_{1}$ is the fused LASSO penalty
(Tibshirani et al., 2005) to control the temporal consistency of the hot-spots.
Traditional algorithms  often involve the storage and computation of the matrix $\mathbf{B}_m$ and $\mathbf{B}_h$, which is of the dimension $n_1n_2n_3 \times n_1n_2n_3.$ Thus they might work to solve the optimization problem in equation \eqref{equ:eatimation}
when the dimensions are small, but they will be computationally infeasible as the dimensions grow. To address this computational challenge, we propose to simplify the computational complexity by modifying the matrix algebra in traditional algorithm into tensor algebra, and will discuss how to optimize the problem in equation \eqref{equ:eatimation} computationally efficiently in Section \ref{sec:estimation}.

\section{Hot-spot Detection} \label{sec:hot-spot_detection}

This section focuses on the detection of the hot-spot, which includes the detection and identification of the year (when), the state (where) and the week (which) of the hot-spots. In our case study, we focus on the upward shift of the number of gonorrhea cases, since the increasing gonorrhea is generally more harmful to the societies and communities.
Of course, one can also detect the downward shift with a slight modification of our proposed algorithms by multiplying $-1$ to the raw data.

For the purpose of easy presentation, we first discuss the detection of the hot-spot, i.e., detect when hot-spot occurs in
Subsection \ref{sec:Temporal_Detection}. Then,  in Subsection \ref{sec:spatial_location}, we consider the localization of the hot-spot, i.e.,
determine which states and which weeks are involved for the detected hot-spots.

\subsection{Detect When the Hot Spot Occurs}
\label{sec:Temporal_Detection}

To determine when the hot-spot occurs, we consider the following hypothesis test and set up the control chart for the hot-spot detection \eqref{eq:chaneg_hypothesis_testing}.
\begin{equation}
H_{0}:
\widetilde{\mathbf{r}}_{t} = 0
\;\;\;
v.s.
\;\;\;
H_{1}:
\widetilde{\mathbf{r}}_{t} = \delta \widehat{\mathbf{h}}_{t} \;\;\; (\delta>0),
\label{eq:chaneg_hypothesis_testing}
\end{equation}
where $\widetilde{\mathbf{r}}_{t}$ is the expected residuals after removing the mean.
The essence of this test is that, we want to detect whether $\widetilde{\mathbf{r}}_{t}$ has a mean shift in the direction of $\widehat{\mathbf{h}}_{t}$, estimated in Section \ref{sec:estimation}.
To test this hypotheses, the likelihood ratio test is applied to the residual $\mathbf{r}_{t}$ at each time $t$, i.e. $\mathbf{r}_{t}=\mathbf{y}_{t}-\boldsymbol{\mu}_{t}$, where it assumes that the residuals $\mathbf{r}_{t}$ is independent after removing the mean and its distribution before and after the hot-spot remains the same.
Accordingly, the test statistics monitoring upward shift is designed as
$
P_{t}^{+}
=
\widehat{\mathbf{h}}_{t}'^{+}\mathbf{r}_{t}/\sqrt{\widehat{\mathbf{h}}_{t}'^{+}\widehat{\mathbf{h}}_{t}^{+}}
$
\citep{hawkins1993regression},
where $\widehat{\mathbf{h}}_{t}^{+}$ only takes the positive part of $\widehat{\mathbf{h}}_{t}$ with other entries as zero.
Here we put a superscript ``+'' to emphasis that it aims for upward shift.

\yujie{
The choices of the penalty parameters $\lambda_{1},\lambda_{2}$ are describled as follows.
}
In order to select the one with the most power, we propose to calculate a series of $P_{t}^{+}$ under different combination of $(\lambda_{1},\lambda_{2})$ from the set
$
\Gamma
=
\{(\lambda_{1}^{(1)},\lambda_{2}^{(1)})\cdots(\lambda_{1}^{(n_{\lambda})},\lambda_{2}^{(n_{\lambda})})\}
$.
For better illustration, we denote the test statistics under penalty parameter $(\lambda_{1},\lambda_{2})$ as $P_{t}^{+}(\lambda_{1},\lambda_{2})$.
The test statistics \citep{LASSO} with the most power to detect the change, noted as $\widetilde{P}_{t}^{+}$, can be computed by
\begin{equation}
\widetilde{P}_{t}^{+}=\max_{(\lambda_{1},\lambda_{2})\in\Gamma}\frac{P_{t}^{+}(\lambda_{1},\lambda_{2})-E(P_{t}^{+}(\lambda_{1},\lambda_{2}))}{\sqrt{Var(P_{t}^{+}(\lambda_{1},\lambda_{2}))}},\label{equ:most_power}
\end{equation}
where $E(P_{t}^{+}(\lambda_{1},\lambda_{2}))$, $Var(P_{t}^{+}(\lambda_{1},\lambda_{2}))$ respectively are the mean and variance of $P_{t}(\lambda_{1},\lambda_{2})$ under $H_{0}$ (e.g. for phase-I in-control samples).

Note that the penalty parameter $(\lambda_{1},\lambda_{2})$ to realize the maximization in equation \eqref{equ:most_power} is generally different under different time $t$.
To emphasize such dependence of time $t$, denote by $(\lambda_{1,t}^{*},\lambda_{2,t}^{*})$ the parameter pair that attains the maximization in equation \eqref{equ:most_power} at time $t$, i.e,
\begin{equation}
(\lambda_{1,t}^{*},\lambda_{2,t}^{*})=\arg\max_{(\lambda_{1},\lambda_{2})\in\Gamma}\frac{P_{t}^{+}(\lambda_{1},\lambda_{2})-E(P_{t}^{+}(\lambda_{1},\lambda_{2}))}{\sqrt{Var(P_{t}^{+}(\lambda_{1},\lambda_{2}))}}.
\label{eq:lambda12}
\end{equation}
Thus, the series of the test statistics for the hot-spot at time $t$ is $\widetilde{P}_{t}^{+}(\lambda_{1,t}^{*},\lambda_{2,t}^{*})$ where $t=1\cdots T$.

With the test statistic available, we design a control chart based on the CUSUM procedure due to the following reasons:
(1) we are interested in detecting the change with the temporal continuity, therefore,  aligns with the objective of CUSUM.
(2) In the view of social stability, we want to keep gonorrhea at a target value without sudden changes, which makes the CUSUM chart is a natural better fit.

To be more specific, in the CUSUM procedure, we compute the CUSUM statistics recursively by $$W_{t}^{+}=\max\{0,W_{t-1}^{+}+\widetilde{P}_{t}^{+}(\lambda_{1,t}^{*},\lambda_{2,t}^{*})-d\},
$$
and $W_{t=0}^{+}=0,$ where $d$ is a constant and can be chosen according to the degree of the shift that we want to detect.
Next, we set the control limit $L$
to achieve a desirable ARL for in-control samples.
Finally, whenever $W_{t}^{+} > L$ at some time $t=t^{*},$ we declare that a hot-spot occurs at time $t^{*}$.

\subsection{Localize Where and Which the Hot Spot Occur?}
\label{sec:spatial_location}

    After the hot-spot $t^*$ has been detected by the CUSUM control chart in the previous section, the next step is to localize  where and which crime rate may account for this hot-spot.
    To do so, we propose to utilize the   vector
    $$
    \widehat{\mathbf{h}}_{\lambda_{1, t^{*}}^{*}, \lambda_{2,t^{*}}^{*}}
    =
    \mathbf{B}_{h}\widehat{\boldsymbol{\theta}}_{h, \lambda_{1,t^{*}}^{*}, \lambda_{2,t^{*}}^{*}}
    $$
    at the declared hot-spot time $t^{*}$ and the corresponding  parameter
    $\lambda_{1,t^{*}}^{*},\lambda_{2,t^{*}}^{*}$
    in equation \eqref{eq:lambda12}.
    For the numerical computation purpose, it  is often easier to directly work with the tensor format of the hot-spot
    $
    \widehat{\mathbf{h}}_{ \lambda_{1,t^{*}}^{*},\lambda_{2,t^{*}}^{*}}
    $,
    denoted as
    $
    \widehat{\mathcal{H}}_{\lambda_{1,t^{*}}^{*}, \lambda_{2,t^{*}}^{*}}
    $,
    which is a tenor of dimension $ n_{1} \times n_{2} \times T $.
    If the $(i,j, t^*)$-th entry in $\widehat{\mathcal{H}}_{\lambda_{1,t^{*}}^{*}, \lambda_{2,t^{*}}^{*}}$ is non-zero, then we declare that there is a  hot-spot for the $j$-th crime rate type in the $i$-th state in $t^*$-th year.

\section{Optimization Algorithm}
\label{sec:estimation}

In this section, we will develop an efficient  optimization algorithm for solving the optimization problem in equation \eqref{equ:eatimation}.
For notion convenience, we adjust the notation above a little bit.
Because $\boldsymbol{\theta}_{m},\boldsymbol{\theta}_{h}$ in equation \eqref{equ:eatimation} is solved under penalty $\lambda_{1}R_{1}(\boldsymbol{\theta}_{h})+\lambda_{2}R_{2}(\boldsymbol{\theta}_{h})$, we change $\boldsymbol{\theta}_{m}$, $\boldsymbol{\theta}_{h}$ into $\boldsymbol{\theta}_{m,\lambda_{1},\lambda_{2}},\boldsymbol{\theta}_{h,\lambda_{1},\lambda_{2}}$ to emphasis the penalty parameter $\lambda_{1}$ and $\lambda_{2}$.
Accordingly, $\boldsymbol{\theta}_{h,0,\lambda_{2}}$ refers to the estimator only under the second penalty $\lambda_{2}R_{2}(\boldsymbol{\theta}_{h})$, i.e,
\begin{equation}
\boldsymbol{\theta}_{h,0,\lambda_{2}}=\arg\min_{\boldsymbol{\theta}_{m},\boldsymbol{\theta}_{h}}\{\Vert\mathbf{e}\Vert_{2}^{2}+\lambda R_{2}(\boldsymbol{\theta}_{h})\}.\label{equ:one_penalty}
\end{equation}
The structure of this section is that, we first develop the procedure of our proposed method in
Subsection \ref{sec:algorithm_procedure} and then gives the computational complexity in Subsection \ref{sec:computational_complexity}.

\subsection{Procedure of Our Algorithm}
\label{sec:algorithm_procedure}

In the optimization problem shown in equation \eqref{equ:eatimation}, there are two unknown vectors, namely $\boldsymbol{\theta}_{m,\lambda_{1},\lambda_{2}}$, $\boldsymbol{\theta}_{h,,\lambda_{1},\lambda_{2}}$.
To simplify the optimization above, we first figure out the close-form correlation between $\boldsymbol{\theta}_{m,\lambda_{1},\lambda_{2}}$ and $\boldsymbol{\theta}_{h,\lambda_{1},\lambda_{2}}$.
Then, we solve the optimization by modifying the  matrix algebra in FISTA\citep{FISTA} into tensor algebra.
The key to realize it is the proximal mapping of $\lambda_{1}R_{1}(\boldsymbol{\theta}_{h,\lambda_1,\lambda_2})+\lambda_{2}R_{2}(\boldsymbol{\theta}_{h,\lambda_1,\lambda_2})$.
To address it, we first aims at the proximal mapping of $\lambda_{2}R_{2}(\boldsymbol{\theta}_{h,0,\lambda_1})$, where SFA via gradient descent \citep{liu2010efficient} is used.
And then the proximal mapping of $\lambda_{1}R_{1}(\boldsymbol{\theta}_{h,\lambda_1,\lambda_2})+\lambda_{2}R_{2}(\boldsymbol{\theta}_{h,\lambda_1,\lambda_2})$ can be solved with a close-form correlation between it and the proximal mapping of $\lambda_{2}R_{2}(\boldsymbol{\theta}_{h,0,\lambda_2})$.

There are three subsections in this section, where each subsection represents one step in our proposed algorithm.

\subsubsection{Estimate the mean parameter}

To begin with, we first simplify the optimization problem in equation \eqref{equ:eatimation}, i.e., figure out the close-form correlation between $\boldsymbol{\theta}_{m,\lambda_{1},\lambda_{2}}$ and $\boldsymbol{\theta}_{h,\lambda_{1},\lambda_{2}}$.

Although there are two sets of parameters $\boldsymbol{\theta}_{m,\lambda_{1},\lambda_{2}}$ and $\boldsymbol{\theta}_{h,\lambda_{1},\lambda_{2}}$ in the model, we note that given $\boldsymbol{\theta}_{h,\lambda_{1},\lambda_{2}}$, the parameter $\boldsymbol{\theta}_{m,\lambda_{1},\lambda_{2}}$ is involved in the standard least squared estimation and thus can be solved in the closed-form solution, see equation \eqref{equ:theta_and_theta_a} in the proposition below.

\begin{prop}
	Given $\boldsymbol{\theta}_{h,\lambda_{1},\lambda_{2}}$, the closed-form solution of $\boldsymbol{\theta}_{m,\lambda_{1},\lambda_{2}}$  is given by:
	\begin{equation}
	\boldsymbol{\theta}_{m,\lambda_{1},\lambda_{2}}=
	(\mathbf{B}_{m}'\mathbf{B}_{m})^{-1}(\mathbf{B}_{m}'y-\mathbf{B}_{m}'\mathbf{B}_{h}\boldsymbol{\theta}_{h,\lambda_{1},\lambda_{2}}).
	\label{equ:theta_and_theta_a}
	\end{equation}
\end{prop}

It remains to investigate how to estimate the parameter $\boldsymbol{\theta}_{h,\lambda_{1},\lambda_{2}}.$
After plugging in \eqref{equ:theta_and_theta_a} into \eqref{equ:eatimation}, the optimization problem for estimating $\boldsymbol{\theta}_{h,\lambda_{1},\lambda_{2}}$ becomes
\begin{equation}
\arg\min_{\boldsymbol{\theta}_{h,\lambda_{1},\lambda_{2}}}
\Vert\mathbf{y}^{*}-\mathbf{X}\boldsymbol{\theta}_{h,\lambda_{1},\lambda_{2}}\Vert_{2}^{2}
+
\lambda_{1}\Vert\boldsymbol{\theta}_{h,\lambda_{1},\lambda_{2}}\Vert_{1}
+
\lambda_{2}\Vert\mathbf{D}\boldsymbol{\theta}_{h,\lambda_1,\lambda_2}\Vert_{1}
,\label{equ:opt_with_2_penalty}
\end{equation}
where $\mathbf{y}^{*}=\left[\mathbf{I}-\mathbf{H}_{m}\right]\mathbf{y}$ , $\mathbf X=\left[\mathbf{I}-\mathbf{H}_{m}\right]\mathbf{B}_{h}$ and $\mathbf{H}_{m}=\mathbf{B}_{m}(\mathbf{B}_{m}'\mathbf{B}_{m})^{-1}\mathbf{B}_{m}'$ is the projection matrix.

Due to the high dimension, we need to develop an efficient  and  precise optimization algorithm to optimize\eqref{equ:eatimation}.
Obviously, \eqref{equ:opt_with_2_penalty} is a typical sparse optimization problem.
However, most of the sparse optimization frameworks focus on optimizing Eq. (7).

\begin{equation}
\label{equ:estimate_theta0lambda2_without_D}
\arg\min_{\boldsymbol{\theta}_{h,0,\lambda_{2}}}
\Vert \mathbf{y}^{*}-\mathbf{X}\boldsymbol{\theta}_{h,\lambda_1,0}\Vert_{2}^{2} +
\lambda_{1}\Vert\boldsymbol{\theta}_{h,\lambda_1,0} \Vert_{1},
\end{equation}
such as \cite{ISTA}, \cite{FISTA}, \cite{glmnet} and so on, where iterative updating rule are used base either on the gradient information or the proximal mapping.
In most cases, the algorithms above works, however,  two challenges occur in our paper:
\begin{enumerate}
	\item When the dimension of $\mathbf{X}$ (of size $n_1n_2T \times n_1n_2T$) become increasingly large, it is difficult for the computer to store and memorize it.
	\item When the penalty term is
	$
	\lambda_1 \Vert \boldsymbol{\theta}_{h,\lambda_1,\lambda_2}\Vert_1
	+
	\lambda_{2}\Vert\boldsymbol{D\theta}_{h,\lambda_1,\lambda_{2}}\Vert_{1}
	$, instead of only
	$
	\lambda_{1}\Vert \boldsymbol{\theta}_{h,\lambda_1,\lambda_{2}}\Vert_{1}
	$,
	direct application of the  proximal mapping  of $\lambda_{1}\Vert \boldsymbol{\theta}_{h,\lambda_1,\lambda_{2}}\Vert_{1}$  is not workable.
\end{enumerate}

Therefore, directly applying these above algorithms(\cite{FISTA}, \cite{ISTA}, \cite{glmnet}) to our case is not feasible.
To extend the existing research, we proposed an iterative algorithm in Algorithm \ref{alg:Tensor_and_FISTA} and we explain the approach to solve the proximal mapping of
$
\lambda_1 \Vert \boldsymbol{\theta}_{h,\lambda_1,\lambda_2}\Vert_1
+
\lambda_{2}\Vert\boldsymbol{D\theta}_{h,\lambda_1,\lambda_{2}}\Vert_{1}
$ in Section \ref{sec:proximal_map}.

\subsubsection{Proximal Mapping} \label{sec:proximal_map}

The main tool we use to solve the optimization problem in equation \eqref{equ:opt_with_2_penalty} is a variation of proximal mapping.
Denote that
$
F(\boldsymbol{\theta}_{h,\lambda_1,\lambda_2})
=
\frac{1}{2} \| \mathbf{y}^*-\mathbf{X} \boldsymbol{\theta}_{h,\lambda_1,\lambda_2} \|_2^2 .
$
And in the $i$-th iteration, the according recursive estimator of $\boldsymbol{\theta}_{h,\lambda_1,\lambda_2}$ is noted as $\boldsymbol{\theta}_{h,\lambda_1,\lambda_2}^{(i)}$.
Besides,an auxiliary variable $\boldsymbol{\eta}^{(i)}$ is introduced to update from $\boldsymbol{\theta}_{h,\lambda_1,\lambda_2}^{(i)}$ to $\boldsymbol{\theta}_{h,\lambda_1,\lambda_2}^{(i+1)}$ through

\begin{eqnarray*}
	\boldsymbol{\theta}_{h,\lambda_1,\lambda_2}^{(i+1)}
	& = &
	\arg\min_{\boldsymbol{\theta} }
	F( \boldsymbol{\eta}^{(i)} ) +
	\frac{\partial}{\partial \boldsymbol{\theta}_{h,\lambda_1,\lambda_2}} F( \boldsymbol{\eta}^{(i)} ) \left( \boldsymbol{\theta} - \boldsymbol{\eta}^{(i)} \right) + \\
	& &
	\lambda_1 \| \boldsymbol{\theta} \|_1 +
	\lambda_2 \| \mathbf{D} \boldsymbol{\theta} \|_1 +
	\frac{L}{2} \| \boldsymbol{\theta} - \boldsymbol{\eta}^{(i)} \|_2^2
	\\
	& = &
	\arg\min_{\boldsymbol{\theta} }
	\left[
	\frac{1}{2} \left[ \boldsymbol{\theta} -\left( \boldsymbol{\eta}^{(i)} -\frac{\partial}{L \partial \boldsymbol{\theta}} F( \boldsymbol{\eta}^{(i)} )  \right)  \right]^2 +
	\lambda_1 \| \boldsymbol{\theta} \|_1+
	\lambda_2 \| \mathbf{D} \boldsymbol{\theta} \|_1\right]\\
	& \triangleq & \pi_{\lambda_2}^{ \lambda_1 }(\mathbf{v})
\end{eqnarray*}
where
$
\mathbf{v}=\boldsymbol{\eta}^{(i)} -\frac{\partial}{L \partial \boldsymbol{\theta}} F( \boldsymbol{\eta}^{(i)} )
$,
$
\boldsymbol{\eta}^{(i)} =
\boldsymbol{\theta}_{h,\lambda_1,\lambda_2}^{(i)} + \frac{ t_{i-2} -1 }{t_{i-1}} (\boldsymbol{\theta}_{h,\lambda_1,\lambda_2}^{(i)}   -  \boldsymbol{\theta}_{h,\lambda_1,\lambda_2}^{(i-1)})
$ and $t_{-1}=t_0=1$, $t_{i+1} = \frac{1+\sqrt{1+4t_i^2}}{2}$

Because it is difficult to solve $\pi_{\lambda_2}^{ \lambda_1 }(\mathbf{v})$ directly, we aim to solve $\pi_{\lambda_2}^{ 0 }(\mathbf{v})$ first.
And proved by \cite{liu2010efficient}, there is a close-form correlation between $\pi_{\lambda_2}^{ \lambda_1 }(\mathbf{v})$ and $\pi_{\lambda_2}^{ 0 }(\mathbf{v})$, which is shown in Proposition \ref{prop:theta_0_lambda2_and_theta_lambda1_lambda2}.

\begin{prop}
	\label{prop:theta_0_lambda2_and_theta_lambda1_lambda2}
	The close form relationship between
	$\pi_{\lambda_2}^{\lambda_1}(\mathbf{v}) $  and
	$\pi_{\lambda_2}^{0}(\mathbf{v})$ is
	\begin{equation}
	\pi_{\lambda_2}^{\lambda_1}(\mathbf{v})
	=
	\mbox{sign}(\pi_{\lambda_2}^{0}(\mathbf{v}))
	\odot
	\max\{|\pi_{\lambda_2}^{0}(\mathbf{v})|-\lambda_{1},0\}
	.\label{equ:corre_between_lambda1_and_lambda2}
	\end{equation}
	where $\odot$ is an element-wise product operator.
\end{prop}

With the proximal mapping function in Proposition \ref{prop:theta_0_lambda2_and_theta_lambda1_lambda2}, we can now develop the algorithm shown in Algorithm \ref{alg:Tensor_and_FISTA}.

%\begin{algorithm}[H]
%	\caption{pseudo-code for the estimation of hot-spots  }
%	\LinesNumbered
%	\KwIn{ row data, Bases, Reguarization matrix,  parameters for estimation}
%	\KwOut{ estimation of hot-spots }
%	\bfseries{initialization}\;  	
%	\While{ in the optimization }{
%		use tensor algebra to solve the proximal mapping under one regularization term $\lambda_2 \| \myD_y\|$.
%		
%		use tensor algebra to solve the proximal mapping under two regularization terms $\lambda_1\| \mytheta_h \| + \lambda_2 \| \myD \mytheta_h \|_1$.
%	}
%\end{algorithm}

\begin{algorithm}[H]
	\caption{Iterative updating based on tensor decomposition }
	\label{alg:Tensor_and_FISTA}
	\LinesNumbered
	\KwIn{
		$\mathbf{y}^*,
		\mathbf{B}_s,
		\mathbf{B}_w,
		\mathbf{B}_y,
		\mathbf{D}_s,
		\mathbf{D}_w,
		\mathbf{D}_y,
		K, L,
		\lambda_1, \lambda_2, L_0,
		M_1, M_2$ }
	\KwOut{$
		\boldsymbol{\theta}_{h,\lambda_1,\lambda_2}$}
	\bfseries{initialization}\;  	
	$\boldsymbol{\Theta}^{(1)} = \boldsymbol{\Theta}^{(0)},
	t_{-1}=1,  t_0=1, L=L_0  $\\
	\For{$i =1 \cdots M_1$}{
		$
		\mathcal{N}^{(i)}  =
		\mathcal{N}^{(i)}  + \frac{ t_{i-2} -1 }{t_{i-1}} (\boldsymbol{\Theta}^{(i)}   -  \boldsymbol{\Theta}^{(i-1)})
		$
		\begin{eqnarray*}
			\mathcal{V} & = &\mathcal{N}^{(i)}  -
			\frac{1}{L}
			\mathcal{N}^{(i)} \times_1 (\mathbf P'_s \mathbf P_s) \times_2 (\mathbf P'_w \mathbf P_w) \times_3 (\mathbf P'_y \mathbf P_y) - \\
			& &\frac{1}{L}  \mathcal{Y}^* \times_1 \mathbf P'_s \times_2 \mathbf P'_w \times_3 \mathbf P'_y
		\end{eqnarray*}
		\For{ $j=0 \cdots M_2$ }{
			\begin{eqnarray*}
				\mathcal{G} ^{(i)} & =&
				\left(
				\mathcal{Z}^{(j)} \times_1 (\mathbf D'_s \mathbf D_s) \times_2 (\mathbf D'_w \mathbf D_w) \times_3 (\mathbf D'_y \mathbf D_y) )
				\right) -\\
				& & \left(  \mathcal{V} \times_1 \mathbf D_s \times_2 \mathbf D_w \times_3  \mathbf D_y   \right)
			\end{eqnarray*}
			$
			\mathcal{Z}^{(j+1)} = P\left( \mathcal{Z}^{(j)} - \mathcal{G}^{(j)}/L \right)
			$
		}
		$
		\pi^0_{\lambda_2}(\mathcal {V})  = \mathcal{V} - (\mathcal{Z}^{(M_2)})  \times_1 \mathbf D_s \times_2 \mathbf D_w \times_3  \mathbf D_y
		$\\
		$
		\pi_{\lambda_2}^{\lambda_1} (\mathcal V) = \mbox{sign}( \pi^0_{\lambda_2}(\mathcal V) ) \odot \mbox{max} \{ \left| \pi^0_{\lambda_2}(\mathcal V) \right| -\lambda_1, 0 \}
		$ \\
		$
		t_{i+1} = \frac{1+\sqrt{1+4t_i^2}}{2}
		$
	}
	$
	\widehat{\boldsymbol{\Theta}}_{h,\lambda_1,\lambda_{2}} = \pi_{\lambda_2}^{\lambda_1} (\mathcal V)
	$\\
	$\widehat{\boldsymbol{\theta}}_{h,\lambda_1,\lambda_{2}}=\mbox{vector}(\widehat{\boldsymbol{\Theta}}_{h,\lambda_1,\lambda_{2}})$
	${\boldsymbol{v}}=\mbox{vector}({\mathcal{V}})$
\end{algorithm}

$\mbox{vector}(\cdot)$ is a function that unfolding a order-3 tensor of dimension $n_1\times n_2 \times n_3$ into a vector $n_1n_2n_3$ .

\subsection{Computational Complexity} \label{sec:computational_complexity}

This section discusses the computational complexity of our proposed algorithm.
Suppose the raw data is structured into a tensor of order three with dimensional $n_1 \times n_2 \times n_3$, then the computation complexity of our propose method is of order $O\left( n_1n_2n_3\max\{n_1,n_2,n_3\}  \right)$ (see Proposition \ref{prop:computational_complexity}).
\begin{prop}
	\label{prop:computational_complexity}
	The computational complexity of Algorithm \ref{alg:Tensor_and_FISTA} is of order
	$
	O
	\left(
	n_1n_2n_3\max\{n_1,n_2,n_3\}
	\right)
	$.
\end{prop}

\begin{proof}
	The main computational load in Algorithm \ref{alg:Tensor_and_FISTA} is on the calculation of $\mathbf{v}$ (line 4), $\mathbf{g}^{(i)}$(line 5) and $\pi_{\lambda_2}^{0} (\mathbf v)$ (line 7).
	We will take the calculation of $\mathbf{v}$ in line 4 in the algorithm as an example.
	To begin with, we focus on the computational complexity of
	\begin{equation}
	\label{equ:computional_complexity_proof_part1}
	\mathcal{N}^{(i)} \times_1 (\mathbf P'_s \mathbf P_s) \times_2 (\mathbf P'_w \mathbf P_w) \times_3 (\mathbf P'_y \mathbf P_y) ).
	\end{equation}
	For better illustration, we denote $\mbox{tensor}(\boldsymbol{\eta}^{(i)})$ as $\mathcal{N}^{(i)}$ and $\mathcal{N}^{(i)} \times_1 (\mathbf P'_s \mathbf P_s)$ as tensor $\mathcal L_1$.
	According to the tensor algebra \citep[Section 2.5]{TensorAlgebra},
	$$
	\mathcal L_1 = \mathcal{N}^{(i)} \times_1 (\mathbf P'_s \mathbf P_s)
	\Longleftrightarrow
	\mathcal L_{1(1)} = \mathbf P'_s \mathbf P_s \mathcal N^{(i)}_{(1)}.
	$$
	Therefore, the computational complexity of equation \eqref{equ:computional_complexity_proof_part1} is the same as two-matrix multiplication with order $n_1 \times n_1$ and $n_1 \times n_1n_2$, which is of order $O\left(  n_1n_2n_3(2n_1-1) \right)$.
	
	After the calculation of $\mathcal L_1$, equation \eqref{equ:computional_complexity_proof_part1} is reduced to
	\begin{equation}
	\label{equ:computional_complexity_proof_part2}
	\mathcal L_1  \times_2 (\mathbf P'_w \mathbf P_w) \times_3 (\mathbf P'_y \mathbf P_y) ).
	\end{equation}
	Similarly, denotes  $\mathcal L_2 = \mathcal L_1  \times_2 (\mathbf P'_w \mathbf P_w)$, then
	$$
	\mathcal L_2 = \mathcal L_1 \times_2 (\mathbf P'_w \mathbf P_w)
	\Longleftrightarrow
	\mathcal L_{2(2)} = \mathbf P'_w \mathbf P_w \mathcal N_{(2)}.
	$$
	Therefore, the computational complexity of equation \eqref{equ:computional_complexity_proof_part2} is the same as two-matrix multiplication with order $n_2 \times n_2$ and $n_2 \times n_1n_3$, which is of order $O\left(  n_1n_2n_3(2n_2-1) \right)$.
	
	After the calculation of $\mathcal L_2$, equation \eqref{equ:computional_complexity_proof_part2} is reduced to
	\begin{equation}
	\label{equ:computional_complexity_proof_part3}
	\mathcal L_2   \times_3 (\mathbf P'_y \mathbf P_y) ).
	\end{equation}
	Similarly, denotes  $\mathcal L_3 = \mathcal L_2  \times_2 (\mathbf P'_y \mathbf P_y)$, then
	$$
	\mathcal L_3 = \mathcal L_2 \times_3 (\mathbf P'_y\mathbf P_y)
	\Longleftrightarrow
	\mathcal L_{3(3)} = \mathbf P'_w \mathbf P_w \mathcal N_{(3)}.
	$$
	Therefore, the computational complexity of equation \eqref{equ:computional_complexity_proof_part2} is the same as two-matrix multiplication with order $n_3 \times n_3$ and $n_3 \times n_1n_2$, which is of order $O\left(  n_1n_2n_3(2n_3-1) \right)$.
	
	By combining all these blocks built above, we conclude that the computational complexity of
	equation \eqref{equ:computional_complexity_proof_part1} is of order $O(n_1n_2n_3\left(\max \{n_1,n_2,n_3\} \right))$.
	
	In the same way, the computational complexity in line 5 and 7 of Algorithm \ref{alg:Tensor_and_FISTA} is also of order $O(n_1n_2n_3\left(\max \{n_1,n_2,n_3\} \right))$. Thus, the computational complexity of Algorithm is of order $O(n_1n_2n_3\left(\max \{n_1,n_2,n_3\} \right))$.
\end{proof}

\section{Simulation}
\label{sec:simulation}

In this section, we conduct simulation studies to evaluate our proposed methodologies by comparing with several benchmark methods in the literature.
The structure of this section is as follows.
We first present the data generation mechanism for our simulations in Subsection \ref{sec:sim_data_generation},
then discuss the performance of hot-spot detection and localization  in Subsection \ref{sec:hot-spot_detection_performance}.

\subsection{Generative Model in Simulation}
\label{sec:sim_data_generation}

In our simulation, at each time index $t (t=1 \cdots T)$,  we generate a vector $\mathbf y_t$ of length $n_{1} n_{2} $ by
\begin{equation}
\label{equ:sim_data_generation}
\mathbf y_{i,t}=(\mathbf B\boldsymbol\theta_t)_i+\delta\mathbbm 1\{t\geq \tau\} \mathbbm 1_i\{i \in S_h\}+\mathbf w_{i,t},
\end{equation}
where  $\mathbf y_{i,t}$ denotes the $i$-th entry in vector  $\mathbf y_t$, $(\mathbf{B}\boldsymbol{\theta}_t)_i$ denotes the $i$-th entry in vector $\mathbf B\boldsymbol\theta_t$, and $\delta$ denotes the change magnitude.
Here  $\mathbbm 1(A)$ is the indicator function, which has the value 1 for all elements of $A$ and  the value 0 for all elements  not in $A$, and $\mathbf w_{i,t}$ is the $i$-th entry in the white noise vector whose entries are independent and follow  $N(0,0.1^{2})$ distribution.

Next, after the temporal detection of  hot-spots, we need to further localize the hot-spots in the sense that we need to find out which state and which week may lead to the occurrence of temporal hot-spot.
Because the baseline methods, PCA and T2, can only realize the detection of temporal changes, we only show the localization of spatial hot-spot by SSR-Tensor, SSD \citep{SSD}, ZQ lasso \citep{LASSO}.
For the anomaly setup,  $\mathbbm 1\{t\geq \tau\} $  indicates that the spatial hot-spots only occur after the temporal hot-spot $\tau$.
This ensures that the simulated hot-spot is temporal consistent.
The second indicator function $\mathbbm 1_i\{i \in S_h\}$ shows that only those entries whose location index belongs set $S_h$ are assigned as local hot-spots.
This ensures that the simulated hot-spot is sparse.  Here we assume the change happens at  $ \tau = 50$ among total $T=100$ years.
And the spatial hot-spots index set  is formed by the combination of states Conn, Ohio, West Va, Tex, Hawaii and week from 1-10 and 41-51.

To match the dimension in the case study, we choose $n_{1}=50,n_{2}=51$.
As for the three terms on the right side of equation  \eqref{equ:sim_data_generation}, they serve for the global trend mean, local sparse anomaly and white noise respectively.
In our simulation, the matrix $\mathbf{B}$ is $\mathbf{B}_{m,s} \otimes \mathbf{B}_{m,w} \otimes \mathbf{B}_{m,y}$ with the same choice as that in Section \ref{sec:model}.

Besides, in each of these two scenarios, we further consider two sub-cases, depending on the value of change magnitude $\delta$
in equation \eqref{equ:sim_data_generation}: one is $\delta = 0.1$ (small shift) and the other is $\delta=0.5$ (large shift).

\subsection{Hot-spot Detection Performance}
\label{sec:hot-spot_detection_performance}

In this section, we compare the performance of our proposed method (denoted as `SSR-tensor') for detection of hot-spot with some benchmark methods.
Specifically, we compare our proposed method with
Hotelling $T^{2}$ control chart \citep{T2} (denoted  as `T2'),
LASSO-based control chart proposed by \cite{LASSO} (denoted as `ZQ LASSO'),
PCA-based control chart proposed by \cite{PCA} (denoted as `PCA')
and SSD proposed by \citet{SSD} (denoted as `SSD').
Note that there are two main differences between our SSR-tensor method and the SSD method in \citet{SSD}.
First, SSR-Tensor has the autoregressive or fussed LASSO penalty in equation \eqref{equ:eatimation} so as to
ensure the temporal continuity of the hot-spot.
Second, SSD uses the Shewhart control chart to monitor temporal changes, while SSR-Tensor utilizes CUSUM instead, which is more sensitive for a small shift.

For the basis choices of  our proposed method, to model the spatial structure of the global trend, we choose $\mathbf{B}_{m,1}$ as the kernel matrix to describe the smoothness of the background, whose $(i,j)$ entry is of value $\exp\{-d^2/(2c^2)\}$ where $d$ is the distance between the $i$-th state and $j$-th state and $c$ is the bandwidth chosen by cross-validation.
In addition, we choose identical matrices for the yearly basis and weekly basis since we do not have any prior information.
Moreover, we use the identity matrix for the spatial and temporal basis of the hot-spots.
For SSD in \citet{SSD}, we will use the same spatial and temporal basis in order to have a fair comparison.

For evaluation, we will compute the following four criteria:
(i) precision, defined as the proportion of detected anomalies that are true hot-spots;
(ii) recall, defined as the proportion of the anomalies that are correctly identified;
(iii) F measure, a single criterion that combines the precision and recall by calculating
their harmonic mean; and
(iv) the corresponding average run length ($\mbox{ARL}_1$), a measure on the average detection delay in the special scenario when the change occurs at time $t=1$.
All  simulation results below are based on $1000$ Monte Carlo simulation replications.

Table \ref{table:simulation_hotspot_detection} shows the merits of our methodology  mainly lies on the higher precision and shorter $\mbox{ARL}_1$. For example, when the shift is very small, i.e., $\delta=0.1$, the $\mbox{ARL}_1$ of our SSR-Tensor method is only 1.6420 compared with 7.4970 of SSD and 9.5890 of ZQ-LASSO.
The reason for SSR-Tensor has shorter $\mbox{ARL}_1$ than that of SSD is that, SSD use Shewhart control chart to detect temporal changes, which make it insensitive for a small shift.
While for SSR-Tensor, it applies the CUSUM control chart, which is capable to detect the shift of small size.
The reason for both SSR-Tensor and SSD have shorter $\mbox{ARL}_1$ than that of ZQ-LASSO, PCA and T2 is that ZQ-LASSO fails to capture the global trend mean.
Yet, the data generated in our simulation has both decreasing and circular  global trend, which makes it hard for ZQ-LASSO to model well.

\begin{table}[t]
	%\tiny
	\scriptsize
	\centering
	\begin{tabular}{c|cccc|cccc}
		\hline
		\multicolumn{1}{c|}{ \multirow{2}{*}{methods} } & \multicolumn{4}{c|}{small shift $\delta=0.1$} & \multicolumn{4}{c}{large shift $\delta=0.5$}\\
		\cline{2-9}
		\multicolumn{1}{c|}{} & precision & recall & F measure & ARL & precision & recall & F measure & ARL
		\\
		\hline
		SSR-tensor
		&\bf{0.0824} &\bf{0.9609} &\bf{0.5217} &\bf{1.6420}
		&\bf{0.0822} &\bf{0.9633} &\bf{0.5228} &\bf{1.0002}  \\
		& (0.0025) & (0.0536) & (0.0270) &(0.7214)
		& (0.0022) & (0.0549) & (0.0277) & (0.0144)  \\
		SSD
		&0.0404 &0.9820 &0.5112 &7.4970
		&0.0412 &1.0000 &0.5206 &1.0000 \\
		&(0.0055) & (0.1330) & (0.0692) &(9.4839)
		&(0.0000)&(0.0000) &(0.0000) &(0.0000) \\
		ZQ LASSO
		&0.0412 &1.000  & 0.5206 & 9.5890
		&0.0412 &1.0000 & 0.5206 & 8.8562 \\
		&(0.0000) & (0.0000) & (0.0000) &(7.5414)
		&(0.0000) & (0.0000) & (0.0000) &(7.1169)\\
		PCA
		& - & - & - & 28.7060
		& - & - & - & 32.0469\\
		& - & - & - &(16.9222)
		& - & - & - &(17.4660) \\
		T2
		& - & - & - & 50.0000
		& - & - & - & 50.0000\\
		& - & - & - &(0.0000)
		& - & - & - &(0.0000)\\
		\hline
	\end{tabular}
	\caption{Scenario 1 (decreasing global trend): Comparison of hot-spot detection under small shift and large shift }
	\label{table:simulation_hotspot_detection}
\end{table}

\section{Case Study} \label{sec:case_study}

In this section, we apply our proposed SSR-tensor model and hot-spot detection/localization method to the weekly gonorrhea dataset in
Section \refeq{sec:data}. For the purpose of comparison, we also consider other benchmark methods mentioned in Section \ref{sec:simulation}), and consider two performance criteria: one is the temporal detection of  hot-spots (i.e., which year it occurs) and the other is the localization of the hot-spots (i.e., which state and which week might involve the alarm).

\subsection{When the temporal changes happen?}

Here we consider the performance on the temporal detection of  hot-spots of our proposed method and other benchmark methods. For our proposed SSR-Tensor method, we build a CUSUM control chat utilizing the test statistic in Subsection
\ref{sec:Temporal_Detection}, which is shown in Figure \ref{fig:control_chart_CUSUM}.
From this plot, we can see that the hot-spots are detected at $10$-th year, i.e., 2016.

\begin{figure}[t]
	\centering
	\includegraphics[width=0.5\textwidth]{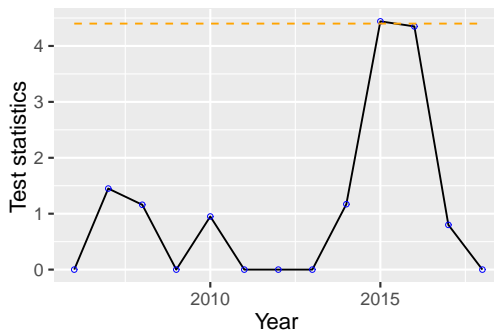}
	\caption{CUSUM Control chart of gonorrhea dataset during years 2006-2018.
	\label{fig:control_chart_CUSUM}}
\end{figure}

For the purpose of comparison, we also apply the benchmark methods, SSD \citep{SSD}, ZQ LASSO \citep{LASSO}, PCA \citep{PCA} and T2\citep{T2}, into the gonorrhea dataset.
Unfortunately, all benchmark methods are unable to raise any alarms, but our proposed SSR-tensor method raises the first hot-spot alarm in year $2016.$

\subsection{Which state and week the spatial hot-spots occur?}

Next, after the temporal detection of  hot-spots, we need to further localize the hot-spots in the sense that we need to find out which state and which week may lead to the occurrence of temporal hot-spot.
Because the baseline methods, SSD, ZQ-LASSO, PCA, and T2, can only realize the detection of temporal changes, we only show the localization of spatial hot-spot by SSR-Tensor, which is visualized in Figure \ref{fig:hot-spot_map_representative}.
\begin{figure}[t]
	\begin{tabular}{ccccc}
		\centering
		\includegraphics[width = 0.19\textwidth]{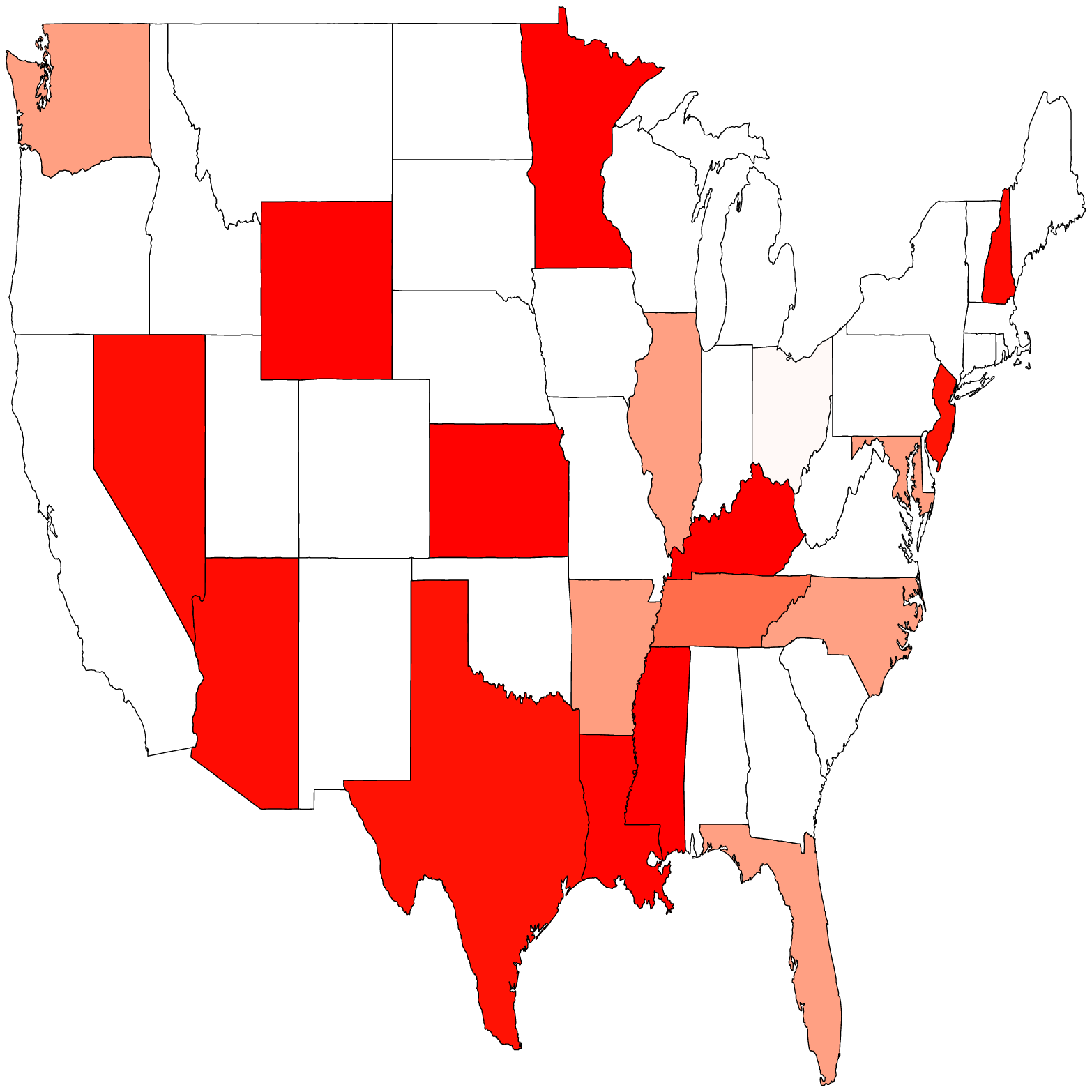}   &
		\includegraphics[width = 0.19\textwidth]{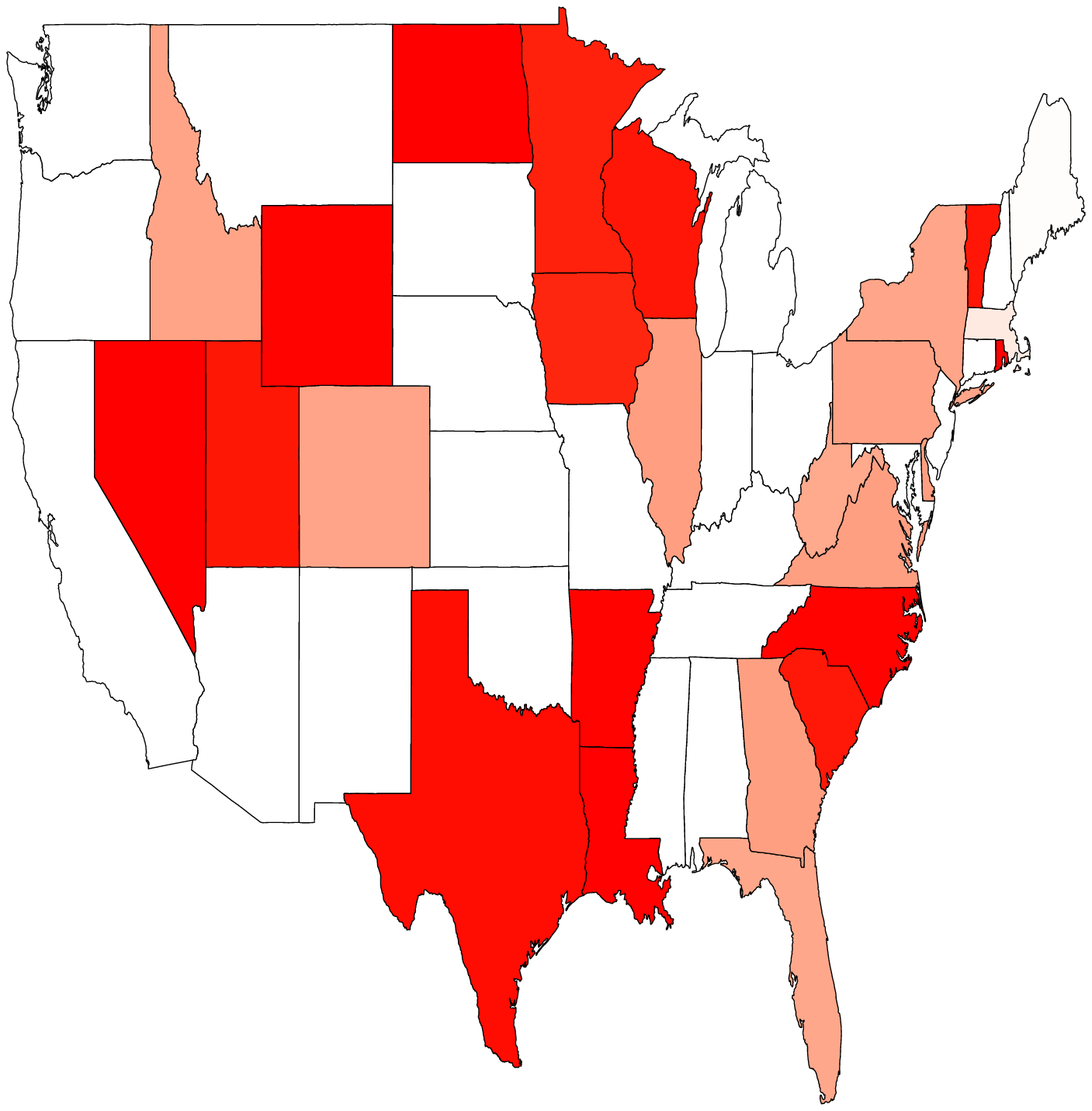}  &
		\includegraphics[width = 0.19\textwidth]{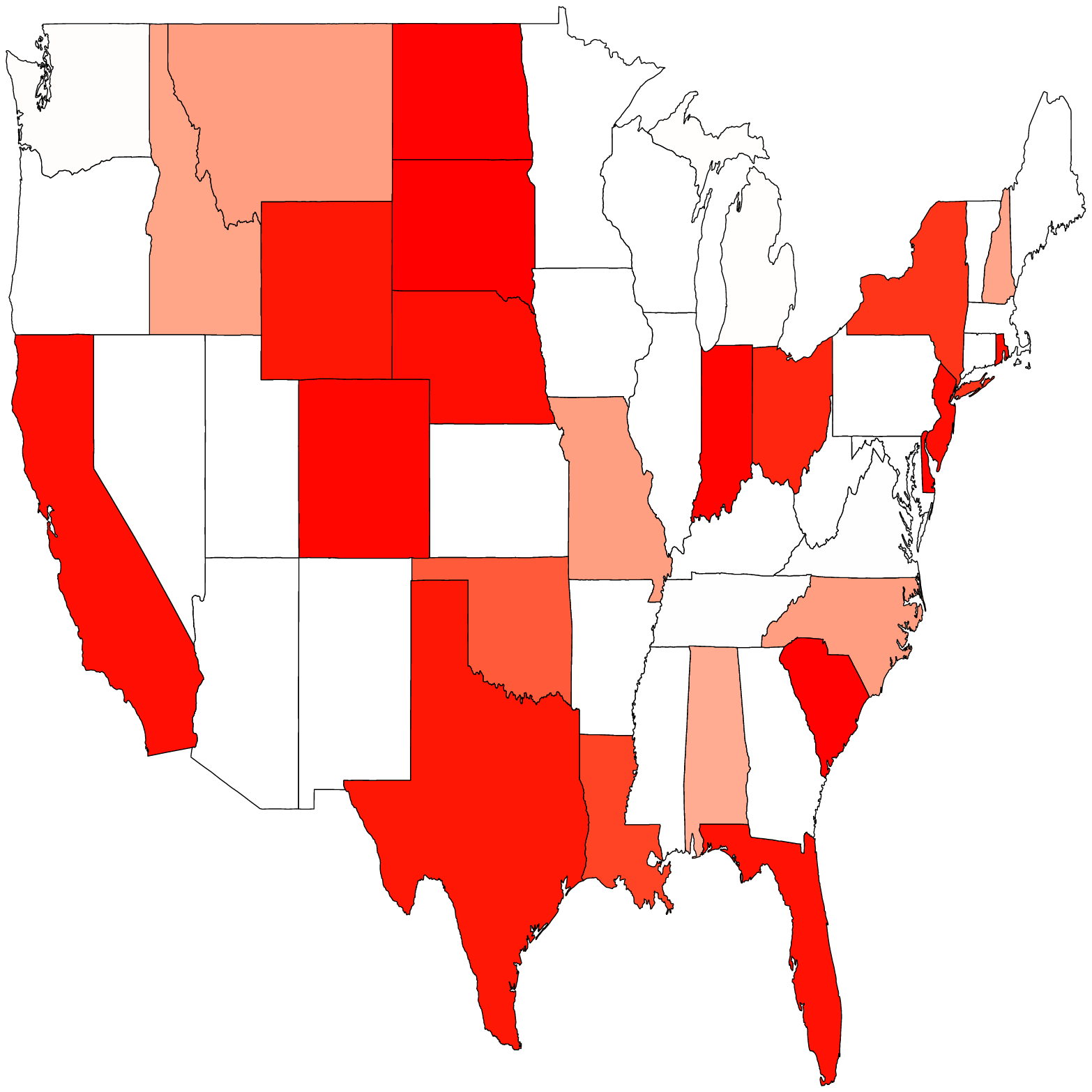}  &
		\includegraphics[width = 0.19\textwidth]{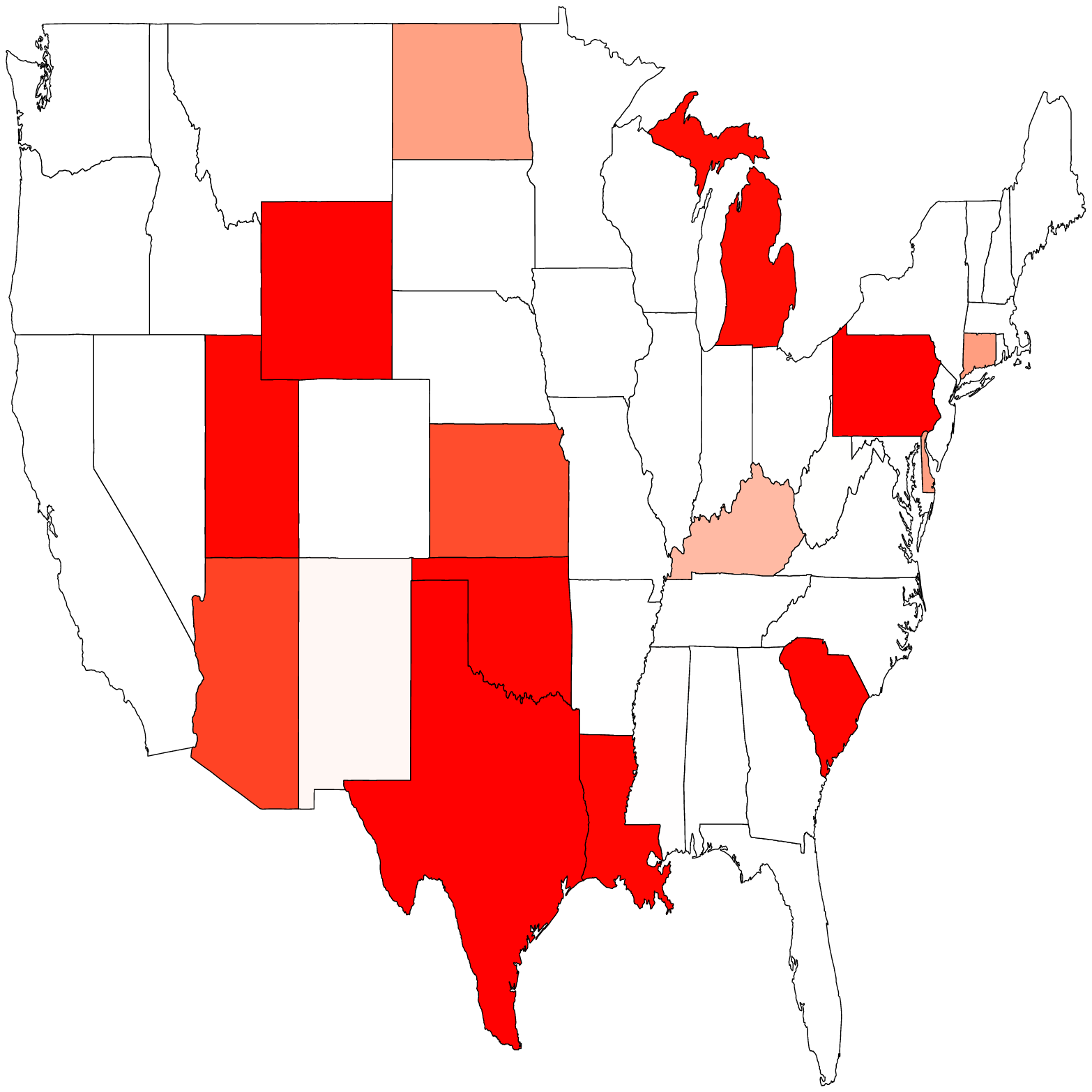}  &
		\includegraphics[width = 0.19\textwidth]{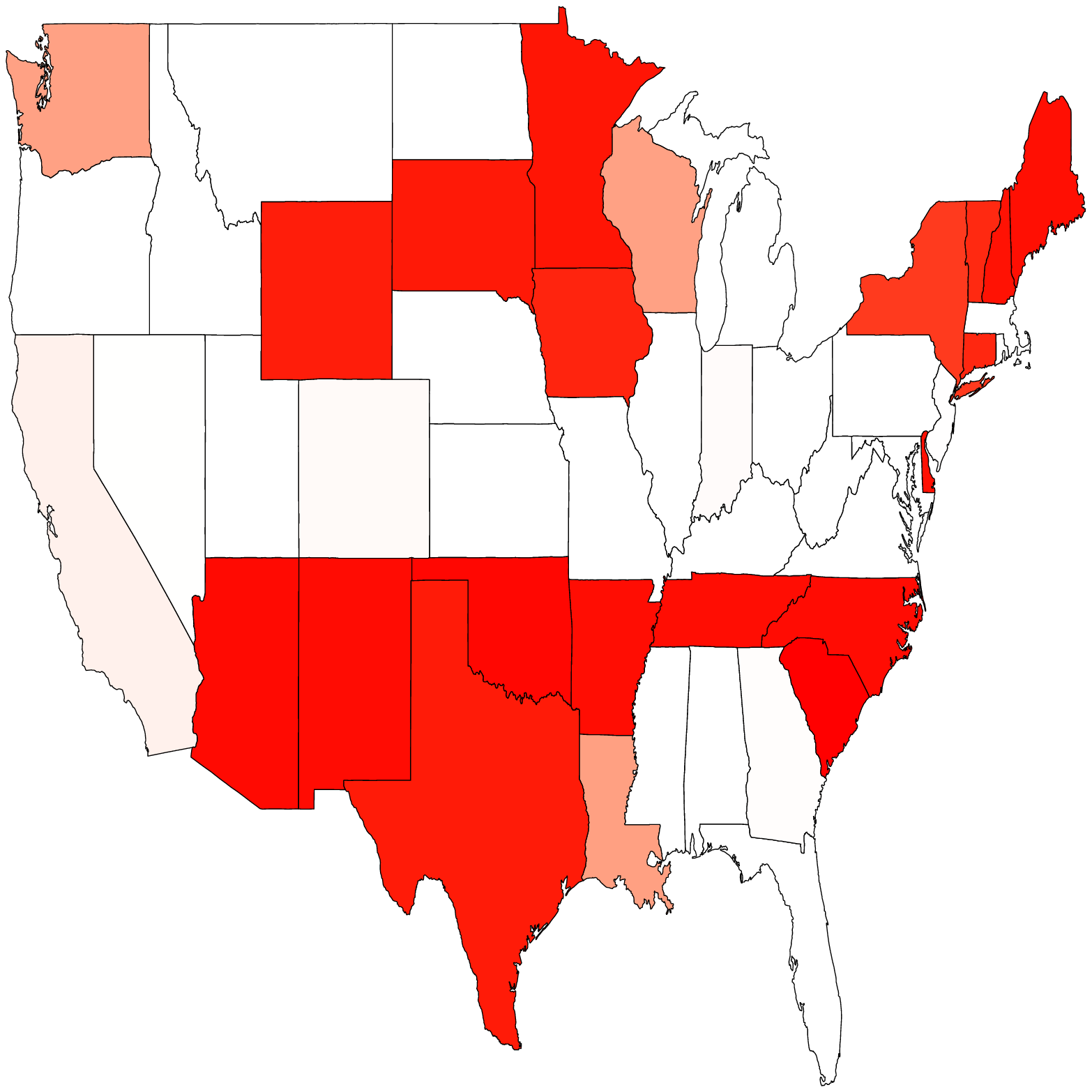} \\
		week 8 & week 19  & week 30 & week 42& week 51
	\end{tabular}
	\caption{Hot-spot detection result of circular pattern of W.S. CENTRAL(Arkansas, Louisiana, Oklahoma, Texas)
		\label{fig:hot-spot_map_representative}	}
\end{figure}

There are some circular patterns in specific areas.
For example, CENTRAL(Ark, La, Okla, Tex) tends to have a circular pattern every $11$ weeks, which is shown
in Figure \ref{fig:hot-spot_map_representative} .
Besides, there are also some circular pattern for a certain state, for instance, Kansas has the bi-weekly pattern as shown in
Figure \ref{fig:acf_and_time_series_for_bi-weekly_pattern}.
To validate the bi-weekly circular pattern of Kansas, we plot the time series plot of Kansas in 2016 as well as the auto-correlation function plot in Figure \ref{fig:hot-spot_map_representative}.
Besides, the auto-correlation function plot in the left panel of
Figure \ref{fig:acf_and_time_series_for_bi-weekly_pattern} serves as a baseline.
It can be seen from the middle and right plot of
Figure \ref{fig:acf_and_time_series_for_bi-weekly_pattern} that, Kansas has some bi-weekly or tri-weekly circular pattern.
\begin{figure}
\centering
	\begin{tabular}{ccc}
		\includegraphics[width = 0.33 \textwidth]{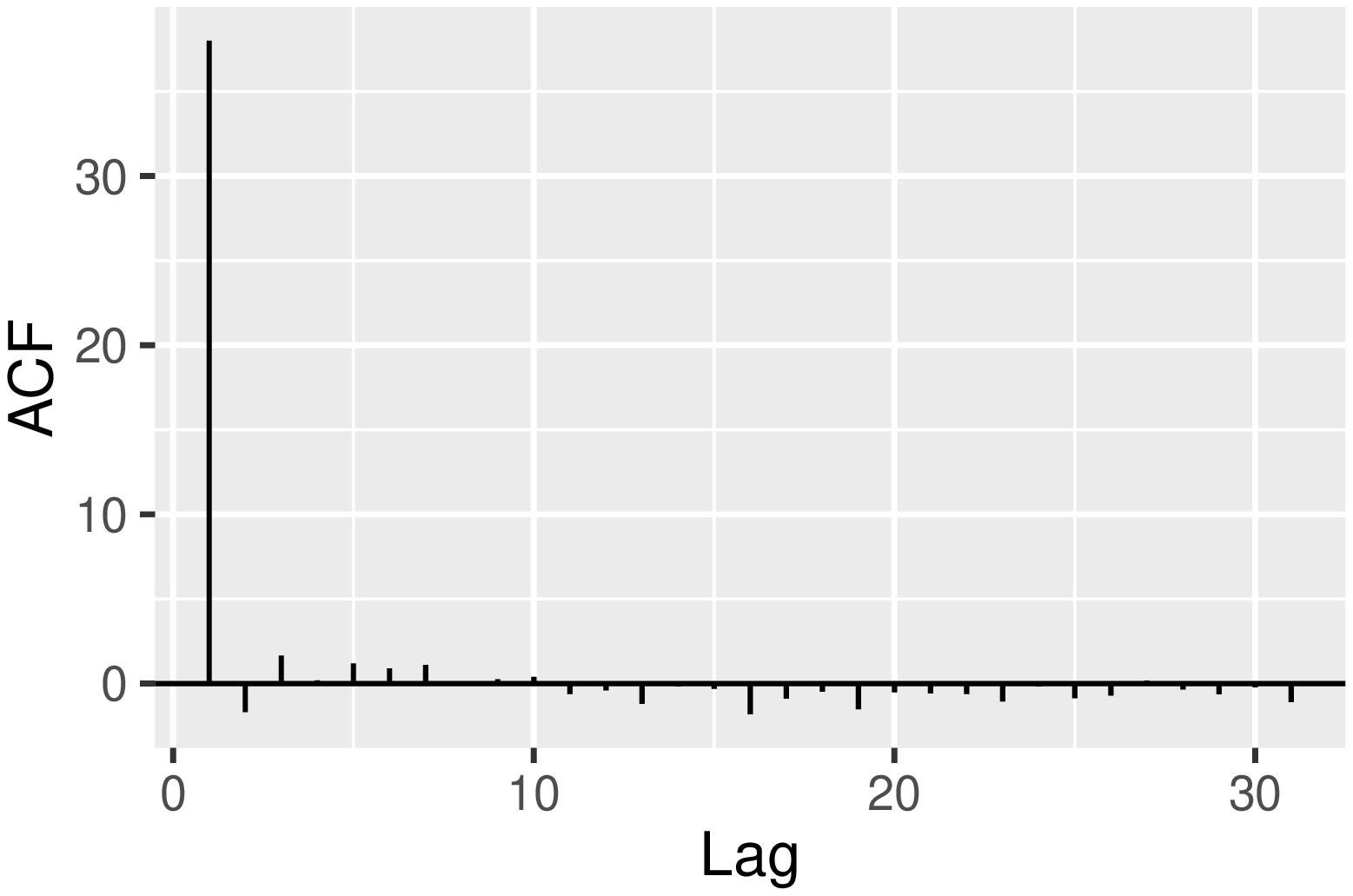} &
		\includegraphics[width = 0.33 \textwidth]{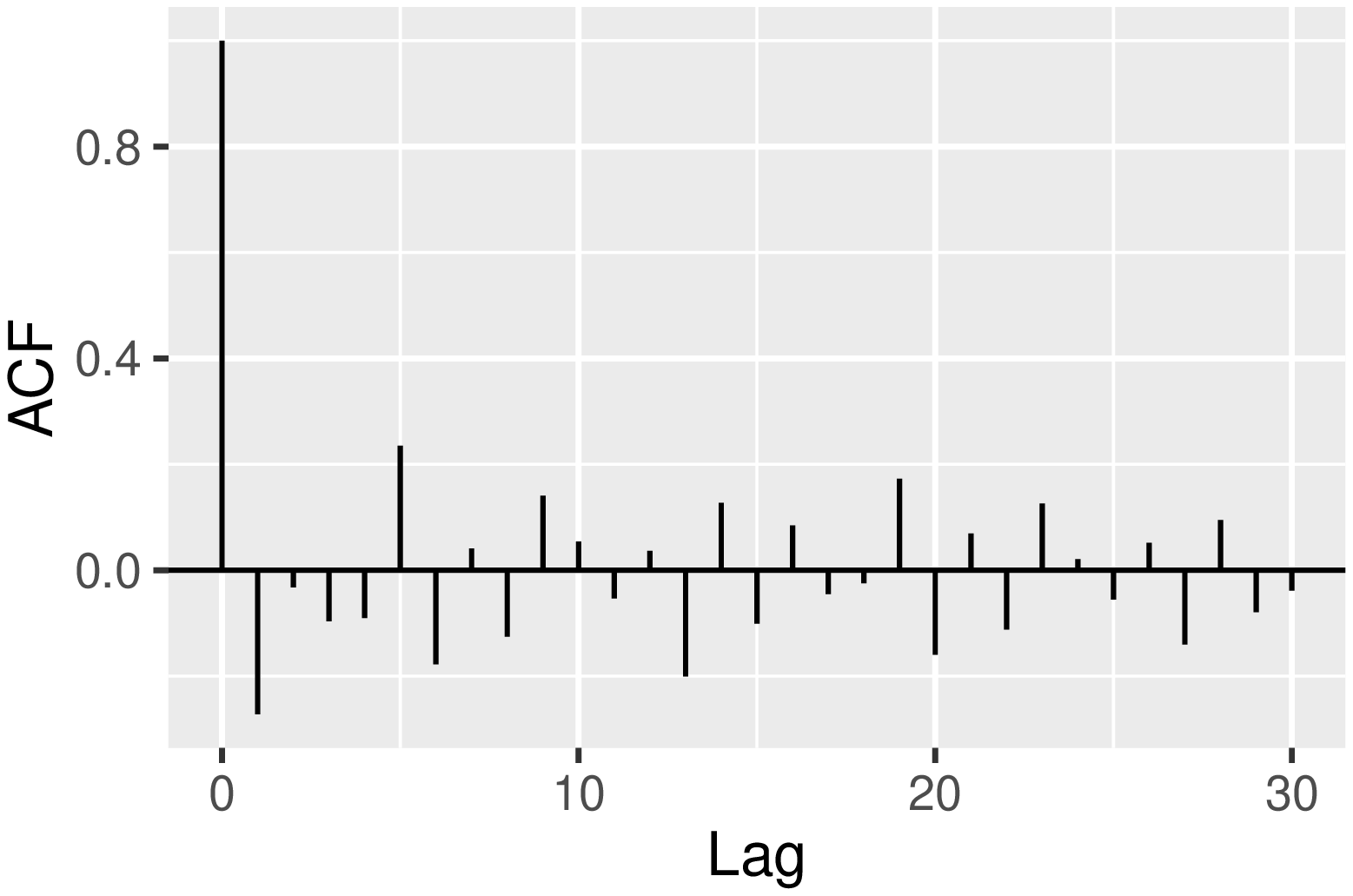} &
		\includegraphics[width = 0.33 \textwidth]{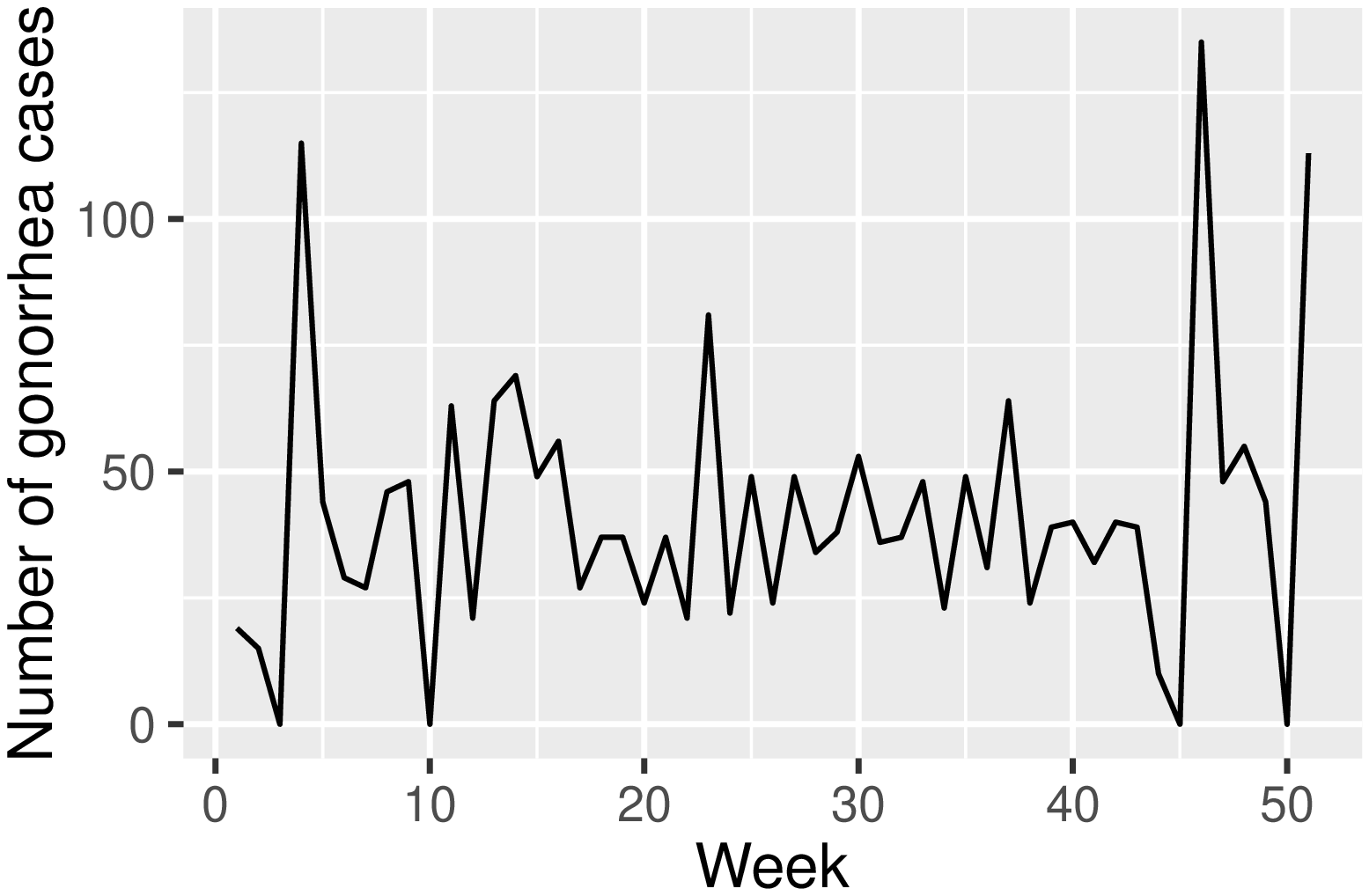}
	\end{tabular}
	\caption{Auto-correlation of all US (left) \& Kans.(middle) in 2016 and time series plot of Kansas in 2016 (right)
		\label{fig:acf_and_time_series_for_bi-weekly_pattern} }
\end{figure}

\newpage
\bibliographystyle{apalike}
\bibliography{reference_genorrhea}

\end{document}